\documentclass[reprint,nofootinbib,amsmath,amssymb,aps, pra]{revtex4-2}
\usepackage{amsmath} % 数式のためのパッケージ
\usepackage{amsthm}%証明のため
\theoremstyle{plain}
\usepackage{stmaryrd}
\usepackage{braket}  % ブラケット記法のためのパッケージ
\usepackage{here}
\usepackage{graphicx}
\usepackage{hyperref}
\newtheorem{thm}{Theorem}
\newtheorem{dfn}{Definition}
\newtheorem{lem}{Lemma}
\usepackage{rotating}
\usepackage{longtable} % longtableを使う際に必要
\usepackage{booktabs} % \toprule, \midrule, \bottomrule を使うため
\usepackage{makecell}
\usepackage{color}

\makeatother % プリアンブルで定義終了

\begin{document}

\title{Efficient Fault-Tolerant Ancilla Preparation for Quantum BCH codes via Cyclic Symmetry.}
% \thanks{Footnote to title of article.}

\author{Kohei Yamamoto}
\email{u660810a@ecs.osaka-u.ac.jp}
\affiliation{%
  Graduate School of Engineering Science, Osaka University, 1-3 Machikaneyama, Toyonaka, Osaka 560-8531, Japan
}

\author{Keisuke Fujii}%
\email{fujii@qc.ee.es.osaka-u.ac.jp}
\affiliation{%
  Graduate School of Engineering Science, Osaka University, 1-3 Machikaneyama, Toyonaka, Osaka 560-8531, Japan
}%
\affiliation{%
  Graduate School of Informatics, Kyoto University, Sakyo-ku, Kyoto, 606-8501, Japan
}%
\affiliation{%
  Center for Quantum Information and Quantum Biology, Osaka University, 1-2 Machikaneyama, Toyonaka 560-0043, Japan
}%
\affiliation{%
  RIKEN Center for Quantum Computing (RQC), Hirosawa 2-1, Wako, Saitama 351-0198, Japan
}%

\begin{abstract}
One of the major challenges in realizing fault-tolerant quantum computers (FTQCs) is the requirement for a large number of physical qubits. To address this issue, high-rate quantum error correcting codes, which efficiently embed logical qubits into physical qubits, have recently attracted considerable attention.
Among such codes, quantum BCH codes, which offer both high rates and large code distances, are promising yet underexplored candidates. However, no fault-tolerant ancilla preparation method specialized for this class had been established.
We employ a two-stage approach (non-fault-tolerant preparation + entanglement distillation) for ancilla preparation. We then propose a framework for designing low-overhead distillation method that strategically leverages the cyclic symmetry of quantum BCH codes to determine which non-fault-tolerant circuits can successfully produce a fault-tolerant state.
Numerical simulations on several high-performance quantum BCH codes up to 127 qubits demonstrate that our method achieves lower spatial overhead and logical error rates than conventional distillation circuits. Furthermore, we evaluated the logical error rates under a circuit-level noise model, and obtained performance benchmarks in realistic settings.
This efficient state preparation technique is expected to contribute to the early realization of practical FTQCs, particularly on highly connected quantum platforms such as neutral atom systems.
\end{abstract}

\maketitle

\section{Introduction}
% 変更後（2026/04/05）
Quantum computers exploit quantum mechanical principles to perform computations and are expected to efficiently solve certain problems that are intractable for classical computers. However, quantum systems are highly susceptible to noise and decoherence, which severely limits the reliability of large-scale quantum computations. Fault-tolerant quantum computation (FTQC)~\cite{FTQC_shor,FTQC_gottesman,NC} addresses this challenge by encoding and protecting quantum information using quantum error correction (QEC).

The surface code has been widely adopted as the standard error correction code due to its simple structure and high error threshold, with extensive efforts dedicated to developing its theoretical framework~\cite{fowler2012surface,fowler2018low,gidney2024magic} and demonstrating its viability on physical hardware~\cite{google2025quantum,google2025msc}.
However, solving practical-scale problems using surface codes requires a massive number of physical qubits, representing one of the greatest barriers to realizing FTQC~\cite{gidney2025factor,zhou2025resource}.

As a promising strategy to address this challenge, high-rate codes, which can encode logical information more densely into physical qubits, have recently attracted significant attention~\cite{breuckmann2021quantum,bravyi2024high}.
Since many high-rate codes possess more complex structures compared to surface codes, their implementation has been considered difficult on platforms with limited connectivity, such as superconducting qubits.
However, with the recent advancements in quantum hardware capable of all-to-all connectivity, such as neutral atom~\cite{bluvstein2022quantum,bluvstein2024logical} and ion trap systems~\cite{Ion_RyanAnderson,Ion_RyanAnderson2,Ion_postler}, the implementation of high-rate codes has become increasingly realistic.

% potentially leading に代わる追加分
High-rate quantum error-correcting codes have been extensively studied through approaches such as concatenated constructions and quantum LDPC codes, yet each faces important limitations. 
Quantum LDPC codes offer the key advantage of maintaining a constant encoding rate asymptotically. However, in the practically relevant regime of moderate block sizes, explicit constructions remain limited, and representative examples such as Bivariate Bicycle Code typically achieve encoding rates of only around 10\%~\cite{bravyi2024high,pecorari2025high,liang2025generalized}. 
Conversely, concatenated high-rate codes can attain significantly higher encoding rates~\cite{goto_hyperqube,yoshida2025concatenate}, but achieving large code distances requires multiple levels of concatenation, leading to a rapid increase in the number of physical qubits and substantial resource overhead.

Motivated by these limitations, we focus on quantum BCH codes~\cite{grassl1999quantum}. 
First introduced by Markus Grassl and Thomas Beth in 1999, this family inherits the favorable properties of classical BCH codes, which are well known for their high code rates and strong error-correcting capabilities~\cite{macwilliams1977theory}. 
In particular, quantum BCH codes can achieve both high encoding rates and large code distances at moderate block sizes~\cite{grassl1999quantum}. 
In this work, we consider a class of such codes whose lengths take the form $2^m-1$, which are well suited for implementation within realistic physical qubit budgets.
Despite these advantages, this class of codes has not been extensively explored. In particular, key ingredients for practical implementations remain largely undeveloped, including fault-tolerant state preparation schemes and systematic performance evaluation under realistic physical error models.

In this paper, we propose an efficient and fault-tolerant state preparation method for logical Pauli basis states of quantum BCH codes.
While various approaches to fault-tolerant state preparation have been studied, including syndrome-measurement-based
~\cite{fowler2012surface,cohen2022low} and flag-based methods~\cite{goto_steane,goto_hyperqube}, we adopt an entanglement distillation approach~\cite{golay,reed_muller} 
in consideration of the structural properties of quantum BCH codes, such as their high-weight stabilizers and generally large code distances.
The proposed construction consists of a non-fault-tolerant encoding circuit followed by an entanglement distillation circuit for purification.
Conventional distillation methods typically require a massive number of ancilla qubits to guarantee fault tolerance, scaling as at least $(t+1)^2$, where $t$ is the number of correctable errors~\cite{gottesman2024surviving}.
To address this issue, we introduce an optimization technique that significantly reduces the required ancilla count by exploiting the cyclic symmetry inherent in quantum BCH codes.
For the $[[31,11,5]]$, $[[63,27,7]]$, and $[[127,71,9]]$ codes, our method reduces the number of required ancilla qubits by more than half compared with conventional approaches.
In addition, to evaluate the performance of the constructed encoding circuit, we perform numerical simulations and compute the output error distribution as well as the success probability.
The simulation results demonstrate reduced error rates at each weight level together with an improved success probability compared to standard methods.
This method provides a systematic framework for resource-efficient fault-tolerant state preparation across the entire quantum BCH code family.

Furthermore, we perform a comprehensive threshold analysis under circuit-level noise using error-corrected teleportation~\cite{knill_gadget} to evaluate the fault-tolerant performance of the proposed scheme.
The study covers all quantum BCH codes with lengths up to 127, as well as length-255 codes with distances up to 13.
In this analysis, we evaluate the performance in terms of the \emph{scaling threshold}, defined as the physical error rate at which the leading-order logical error rate equals the physical error rate.
This metric provides a stricter benchmark than the pseudo-threshold, as it captures asymptotic logical-error suppression at smaller physical error rates.
Our results show that higher optimization levels lead to a reduced logical error rate, resulting in a corresponding improvement in the scaling threshold.
Notably, our systematic evaluation quantitatively clarifies the trade-off between code rate and error-correction performance across a diverse range of quantum BCH code parameters. 
As a representative example, we identify specific codes that maintain scaling thresholds comparable to the Steane code~\cite{Steane_code_1,css_2} while achieving two to four times higher code rates at larger code distances.
Overall, this work provides a systematic and comprehensive threshold evaluation of quantum BCH codes based on explicit fault-tolerant state preparation and error-correction protocols.
These results establish quantum BCH codes as a promising platform for large-scale fault-tolerant quantum computation, offering both high resource efficiency and strong error-correction capability.

\section{Preliminaries: Quantum BCH Codes}
Quantum BCH codes are a subclass of CSS codes~\cite{css_1,css_2} constructed from classical BCH codes.
In this section, we first describe the definition of cyclic codes, which is essential for constructing classical BCH codes. Subsequently, we introduce classical BCH codes and quantum BCH codes.
Although these codes are generally defined over $q$-ary fields, we limit our discussion to the binary field ($q=2$) in this paper.

\subsubsection{Cyclic Codes}
A cyclic code is defined as follows.
\begin{dfn}[Cyclic codes~\protect\cite{daisuukei}]
Let $C$ be a binary linear code with parameters $[n,k]$.
If, for every codeword $c=(c_{n-1},\ldots,c_1,c_0)\in C$, its cyclic shift
$c'=(c_{n-2},\ldots,c_0,c_{n-1})$ also belongs to $C$, then $C$ is called a cyclic code.
\end{dfn}

For algebraic convenience, codewords are represented as polynomials in the quotient ring
$\mathbb{F}_2[x]/(x^n-1)$.
The codeword $c=(c_{n-1},\ldots,c_1,c_0)\in\mathbb{F}_2^n$ is identified with the code polynomial
\begin{equation}
    c(x)=c_{n-1}x^{n-1}+\cdots+c_1x+c_0.
\end{equation}
Under this representation, a cyclic shift corresponds to multiplication by $x$ modulo $x^n-1$:
\begin{equation}
    c'(x)=c_{n-2}x^{n-1}+\cdots+c_0x+c_{n-1}
    \equiv x c(x)\pmod{x^n-1}.
\end{equation}

We next introduce the generator polynomial, which plays a central role in the theory of cyclic codes.
For an $[n,k]$ cyclic code over $\mathbb{F}_2$, the generator polynomial $g(x)$ is defined as the unique monic polynomial of minimum degree among all nonzero code polynomials.
The generator polynomial satisfies the following well-known properties~\cite{daisuukei}:
\begin{enumerate}
\renewcommand{\labelenumi}{(\roman{enumi})}
    \item $g(x)$ is unique.
    \item Every code polynomial $c(x)\in C$ is divisible by $g(x)$.
    \item $g(x)$ divides $x^n-1$ in $\mathbb{F}_2[x]$.
    \item The code dimension is given by $k=n-\deg g(x)$.
\end{enumerate}

Hence, a cyclic code is uniquely characterized by its generator polynomial.
Conversely, any monic polynomial $g(x)\in\mathbb{F}_2[x]$ of degree $s$ that divides $x^n-1$ defines an $[n,n-s]$ cyclic code over $\mathbb{F}_2$.
This correspondence provides a convenient and systematic framework for constructing cyclic codes.

Finally, we describe the parity-check matrix of an $[n,k]$ cyclic code.
Let $g(x)$ be the generator polynomial.
The polynomial $h(x)$ defined by
\begin{equation}
    g(x)h(x)=x^n-1
\end{equation}
is called the check polynomial and can be written as
\begin{equation}
    h(x)=\frac{x^n-1}{g(x)}=h_k x^k+\cdots+h_1 x+h_0 .
\end{equation}
The parity-check matrix $H$ is constructed from the coefficients of $h(x)$ as a circulant matrix,
\begin{equation}
\label{eq:cyclic_pcm}
H=\begin{bmatrix}
    h_0 & h_1 & h_2 & \cdots & h_{k} & 0 & 0 & \cdots &0\\
    0 & h_0 & h_1 & h_2 & \cdots & h_{k}&0& \cdots &0\\
    0& 0 & h_0 & h_1 & h_2 & \cdots & h_{k}& \cdots &0\\
    \vdots & \vdots & \vdots   & \cdots & \vdots & \vdots & \vdots & \cdots &\vdots\\
    0& 0 & 0   & \cdots & &  &  & \cdots &h_k\\
    \end{bmatrix}.
\end{equation}

\subsubsection{Classical BCH Codes}
We begin by defining classical BCH codes.

\begin{dfn}[BCH codes~\cite{daisuukei}]
Let $\alpha$ be a primitive element of the finite field $\mathbb{F}_{2^m}$, namely, a generator of the multiplicative group $\mathbb{F}_{2^m}^{\times}$.
For each $\beta \in \mathbb{F}_{2^m}$, let $m_{\beta}(x)\in\mathbb{F}_2[x]$ denote the minimal polynomial of $\beta$ over $\mathbb{F}_2$, i.e., the monic polynomial of smallest degree having $\beta$ as a root.
Define
\begin{equation}
    g(x) = \mathrm{LCM} \bigl\{ m_{\alpha}(x), m_{\alpha^2}(x), \dots, m_{\alpha^{\delta-1}}(x) \bigr\},
\end{equation}
where $2\le \delta \le n$ and $n=2^m-1$.
Then the cyclic code of length $n$ generated by $g(x)$ is called a BCH code, specifically a narrow-sense primitive BCH code, with designed distance $\delta$.
\end{dfn}

We now state the BCH bound, which describes a fundamental property regarding the minimum distance of BCH codes.
Let $\delta$ be the designed distance and $d_{\min}$ be the minimum distance of a BCH code.
The following relationship holds:
\begin{equation}
    d_{\min} \ge \delta.
\end{equation}
This property guarantees that a BCH code constructed with a designed distance $\delta$ achieves a minimum distance of at least $\delta$, thereby ensuring the desired error-correction capability.

Moreover, BCH codes are highly resilient to random errors and exhibit superior code rates in the short- to intermediate-length regime. This efficiency is evidenced by the fact that many BCH codes with lengths up to $512$ and distances up to $29$ are classified as ``best codes known" in the tables compiled by MacWilliams and Sloane~\cite{macwilliams1977theory}. These are codes that maximize the number of logical bits $k$ for a given $n$ and $d$, and such excellent classical parameters provide a strong motivation for exploring their quantum counterparts to achieve high-performance quantum error correction.

\subsubsection{Quantum BCH Codes}
Quantum BCH codes are CSS codes constructed based on classical BCH codes.

\begin{dfn}[Quantum BCH codes~\cite{grassl1999quantum}]
Let $C$ be a classical BCH code with parameters $[n_c, k_c, d_c]$ and parity check matrix $H_c$.
If the classical code $C$ satisfies the dual-containing constraint (i.e., $C^\perp \subset C$), a CSS code can be constructed using $C$ and its dual code $C^{\perp}$.
The stabilizer matrix of the resulting quantum BCH code $\llbracket n=n_c, k=2k_c-n, d=d_c \rrbracket$ is given by
\begin{equation}
    H=
    \begin{pmatrix}
    H_c & 0 \\
    0 & H_c \\
    \end{pmatrix}.
\end{equation}
\end{dfn}

Quantum BCH codes inherit the characteristics of classical BCH codes, possessing high code rates and good error-correction capabilities.
For example, among codes with comparable code lengths and code distances, the quantum BCH code $\llbracket 127,43,13 \rrbracket$ achieves an encoding rate approximately $50$ times higher than that of the surface code $\llbracket 144,1,12 \rrbracket$ and about $4$ times higher than that of the BB code $\llbracket 144,12,12 \rrbracket$.
However, the stabilizer operators of quantum BCH codes exhibit a complex structure.
For instance, the stabilizer operators of the quantum BCH $\llbracket 31,11,5 \rrbracket$ code are 12-body operators (operators with a weight of 12). In general, the stabilizer operators of quantum BCH codes tend to be high-weight operators.

\section{Fault-tolerant ancilla preparation for quantum BCH codes}
\label{sec:main}
In this section, we present our fault-tolerant (FT) state preparation method for the logical zero state $\ket{0_L^{\otimes k}}$ of quantum BCH codes.
We begin by defining fault tolerance and then review the state preparation method based on entanglement distillation.
Finally, we introduce our main contribution, cyclic symmetry can be exploited to improve the efficiency of entanglement-distillation protocols for quantum BCH codes.

\begin{dfn}[Strict Fault Tolerance~\cite{golay}]
Consider a quantum error-correcting code with distance \(d\). A state-preparation circuit is strictly fault-tolerant (strict-FT) if, for any \(w\) faults with \(w \le \lfloor d/2 \rfloor\), the reduced weight of the output error is at most \(w\).
\end{dfn}

Here, the reduced weight is the minimum weight up to the freedom of stabilizer and logical operators that stabilize the logical state.
For example, even if a fault propagates to produce a high-weight physical error, it is regarded as reduced weight zero if it is equivalent to a stabilizer operator. 
Thus, the strict-FT condition may still be satisfied.

For CSS codes, a non-fault-tolerant (non-FT) preparation circuit for $\ket{0}^{\otimes n}$ can be constructed using CNOT and Hadamard gates~\cite{mqt_qecc}. 
However, such circuits typically involve a large number of CNOT gates, which precludes them from satisfying the strict-FT condition. 
To transform these noisy states into high-fidelity logical states, we introduce an entanglement distillation protocol.

A distillation circuit generally consists of two steps: the first step detecting $X$ errors, and the second step detecting $Z$ errors.
In particular, each step corresponds to a classical linear code.
In this work, we adopt the simplest repetition-code-based $m$-to-1 protocol using the $[m,1,m]$ classical repetition code.
It is known that for a stabilizer code $\llbracket n,k,d=2t+1 \rrbracket$, $(t+1)$-to-1 repetition protocol ensures strict-FT ~\cite{gottesman2024surviving}.

Suppose we aim to reduce the number of ancilla blocks by employing a $t$-to-1 protocol (with fewer ancillas than $(t+1)$).
A major difficulty is the existence of malignant fault patterns: if faults occur at the same physical positions within the non-FT preparation circuits of different ancillas, these errors may cancel out during distillation, preventing correct error detection.
Recent works have shown that code symmetries can be exploited to mitigate such malignant fault alignments, leading to more efficient state preparation protocols.
For example, in quantum Golay~\cite{golay} and Reed–Muller codes~\cite{reed_muller}, specific symmetry structures are used to permute error patterns among ancillas so that their supports are mapped to distinct coordinate positions.
In contrast, for quantum BCH codes, the structure of code symmetries has not been explicitly characterized, and their systematic use for suppressing malignant fault patterns has not been developed.
We address this by identifying and formalizing the relevant symmetries, and establishing a symmetry-based design principle that enables a systematic construction of efficient state preparation protocols across the quantum BCH code family
(as illustrated in Fig.~\ref{fig:method_overview}).
\begin{figure}[t] % figure* にして [t] を指定
  \centering
  \includegraphics[width=0.5\textwidth]{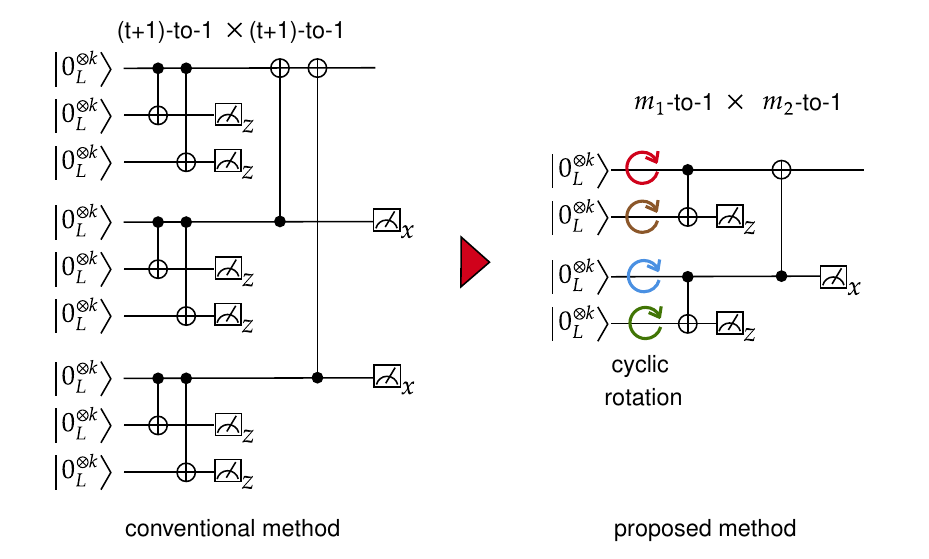} % width=\textwidth でOK（ページ幅いっぱいになる）
  \caption{Schematic overview of the proposed method. An example for the case $t=2$ and $m_1=m_2=2$ is shown.}
  \label{fig:method_overview}
\end{figure}

Quantum BCH codes are CSS codes constructed from classical binary BCH codes, inheriting their permutation symmetries from the automorphism group of the underlying classical code.
For classical binary BCH codes, the following result is established:
\begin{thm}[{~\cite[Theorem 8]{BCH_automorphism}}]
For a binary BCH code with parameters $n=2^m-1$,
\begin{equation}
\label{eq:BCHAut}
\text{Aut}(\text{BCH}(n,k,d)) \supseteq C_n \rtimes F_m,
\end{equation}
where $C_n$ denotes the cyclic permutation group and $F_m$ the Frobenius permutation group.
\end{thm}
This inclusion is generally tight; equality holds for most parameters, with notable exceptions being the Reed–Muller family and specific trivial cases. Detailed characterizations of these automorphism groups can be found in ref.~\cite{BCH_automorphism}.
Here, the cyclic group $C_n$ consists of coordinate shifts $j \mapsto j+1 \pmod n$, while the Frobenius group $F_m$ is generated by the permutation $j \mapsto 2j \pmod n$ corresponding to the field automorphism.

Consequently, every quantum BCH code inherits symmetries from the semi-direct product $C_n \rtimes F_m$, which we broadly denote as ``cyclic symmetry'' (including the Frobenius maps).
Specifically, the actions of the cyclic shift operator $R$ and the Frobenius shift operator $F$ on the physical qubits constitute the logical operators themselves.
As a result, the encoded logical states $|0_L^{\otimes k}\rangle$ and $|+_L^{\otimes k}\rangle$ form eigenstates for these operators.
We exploit this inherent code symmetry to efficiently permute error patterns within the entanglement distillation procedure.

The performance of the proposed method depends on the structure of the underlying non-fault-tolerant (non-FT) circuits.
To establish a realistic baseline with minimized gate counts, we employed the \texttt{mqt.qecc} compiler~\cite{mqt_qecc}, which demonstrated superior performance compared to other tools (Qiskit~\cite{qiskit2024}, Pytket~\cite{Pytket}, Stim~\cite{stim}, A* algorithm~\cite{astar}, and PyZX~\cite{kissinger2019pyzx}) in our preliminary benchmarks.
The specific non-FT circuit structures used in this study are provided in Appendix~\ref{sec:appendix_circuit}.

We next describe how to search for pairs of cyclic permutations that satisfy strict fault tolerance.
The fault-tolerance condition can be checked separately for $X$ errors detected in the first step and $Z$ errors detected in the second step.
It is therefore sufficient to verify that no malignant pattern exists in either case.
Specifically, we consider all $w$-fault errors with $2\leq w\leq \lfloor d/2 \rfloor$.
For each $w$, we divide the analysis into cases according to how the $w$ faults are distributed among the non-FT preparation circuits and, when necessary, faults associated with transversal CNOT gates.
Each case is then checked separately.

We now explain the malignant-pattern check in detail using the quantum BCH code $\llbracket127,71,9\rrbracket$ as an example.
The malignant patterns that must be checked for $X$ and $Z$ errors are illustrated in Figs.~\ref{fig:ft_check}(a) and \ref{fig:ft_check}(b), respectively.
We focus on the $X$-error case in Fig.~\ref{fig:ft_check}(a).
Specifically, we consider the case $w=3$, where two faults occur in the first non-FT preparation circuit and one fault occurs in the second.
First, for each non-FT preparation circuit, we simulate all possible single-fault errors using Stim.
For each resulting error, after applying the corresponding cyclic permutation, we compute its parities with respect to the stabilizer generators and logical operators, and store them in a table.
We then use the fact that any multi-fault error can be represented as the sum of single-fault errors.
Let $s_1$ be the syndrome of a two-fault error from the first non-FT preparation circuit, and let $s_2$ be the syndrome of a one-fault error from the second non-FT preparation circuit.
The condition for the error to be undetected in the first step is $s_1+s_2=0$, or equivalently $s_1=s_2$.
Thus, undetected error patterns can be enumerated by finding pairs of errors with matching syndromes.
A naive enumeration of all three-fault patterns is memory intensive.
Since the number of single-fault errors is of order $10^3$, roughly twice the number of CNOT gates in the non-FT preparation circuit, a naive construction would require a table with on the order of $10^9$ rows.
To avoid this, we formulate the search as a matching problem for equal syndromes and use a meet-in-the-middle strategy\cite{reed_muller}.
Finally, for each undetected pattern, we compute the reduced weight of the remaining $X$ error.
If there exists an undetected pattern whose reduced weight is at least $w+1$, the pattern is malignant.
Therefore, the chosen pair of cyclic permutations satisfies the strict-FT condition if no such pattern exists.
The same procedure can be applied to the other $X$-error patterns and to the $Z$-error patterns.

The possible distributions of $w$ faults are classified based on the following considerations.
In principle, faults may occur in the non-FT preparation circuits, transversal CNOT gates, and transversal measurements.
However, faults in transversal CNOT gates and measurements usually result in error patterns that are already covered by the preparation-circuit faults, or they do not lead to error amplification.
Therefore, in most cases, it is sufficient to classify fault patterns only by how the faults are distributed among the non-FT preparation circuits.
Additional cases need to be considered only for an $m$-to-$1$ protocol with $m\geq 3$.
In this case, errors on the upper block occurring before one of the first through $(m-1)$-th transversal CNOT gates can propagate through the final transversal CNOT and be amplified.
In the present example, such errors correspond to the starred location in Fig.~\ref{fig:ft_check}(b).
We therefore additionally check that there is no malignant pattern satisfying
$s_1=s_2=s_3+t_1$
whose remaining error, associated with $s_1+t_1$, has reduced weight at least $5$.

\begin{figure}[t] % figure* にして [t] を指定
  \centering
  \includegraphics[width=0.5\textwidth]{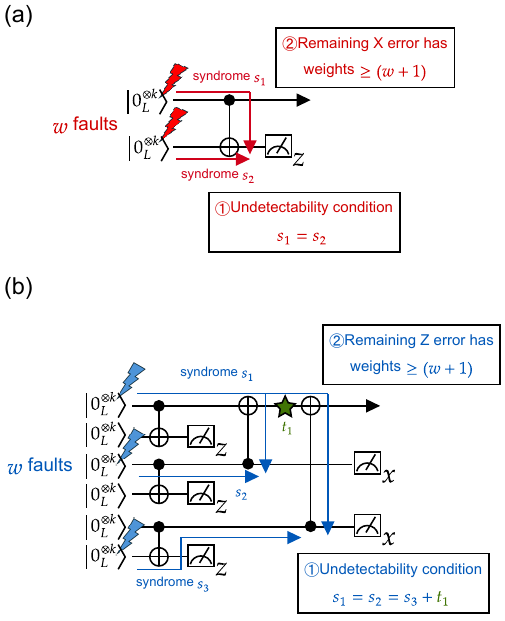} % width=\textwidth でOK（ページ幅いっぱいになる）
  \caption{Example of the strict-FT check for the quantum BCH code $\llbracket127,71,9\rrbracket$.
All $w$-fault errors with $2\leq w\leq 4$ are considered.
(a) Malignant-pattern check for $X$ errors in the first step using a 2-to-1 protocol.
(b) Malignant-pattern check for $Z$ errors in the second step using a 3-to-1 protocol.
Starred locations denote transversal-CNOT errors that require additional checks because they can be amplified by the final transversal CNOT.}
  \label{fig:ft_check}
\end{figure}

Applying this analysis to the quantum BCH codes $[[31,11,5]]$, $[[63,27,7]]$, and $[[127,71,9]]$, we observed that their inherent cyclic symmetry completely eliminates malignant fault patterns (see Table \ref{table:QBCHEfficiency}). 
This elimination allows for strict-FT distillation using significantly fewer ancilla blocks than the conventional protocol. 
Notably, we succeeded in drastically reducing the required ancilla count to $4/9$, $4/16$, and $6/25$, respectively, demonstrating the distinct advantage of the quantum BCH family in our optimization technique.

Although our simulations use CNOT-efficient non-FT circuits generated by the \texttt{mqt.qecc} compiler, the proposed optimization framework is applicable to general non-FT preparation circuits.
Further reductions in the CNOT count or circuit depth of the underlying circuits may enable additional levels of optimization, leading to further reductions in the required ancilla overhead.

\begin{table*}[t]
    \centering
    \begin{tabular}{lllll}
    \hline\hline
        Code & Non-FT & Standard & Optimized & \multicolumn{1}{c}{Example} \\
        $[[n,k,d]]$ & CNOT count& Protocol & Protocol & \multicolumn{1}{c}{Operator Sequence} \\
        \hline
        $[[31, 11, 5]]$  & 73 & $3$-to-$1 \times 3$-to-$1$ & $2$-to-$1 \times 2$-to-$1$ & $((I, R^6), (R^{12}, F))$\\
        $[[63, 27, 7]]$  & 189 & $4$-to-$1 \times 4$-to-$1$ & $2$-to-$1 \times 2$-to-$1$ & $((I,R^8),(R^{24},R^{32}))$ \\
        $[[127, 71, 9]]$ & 594 & $5$-to-$1 \times 5$-to-$1$ & $2$-to-$1 \times 3$-to-$1$ & $((I, R^{15}), (R^{30}, R^{45}), (R^{60}, R^{75}))$\\
        \hline\hline
    \end{tabular}
    \caption{
    Comparison of distillation protocols and symmetry configurations for quantum BCH codes.
    The ``Standard Protocol'' column indicates the conventional requirement based on the code distance (i.e., $(t+1)$-to-1), while the ``Optimized Protocol'' shows the reduced resource requirement achieved using our symmetry-based method.
    The last column provides an example of the symmetry operators (cyclic shifts $R$ and Frobenius shifts $F$).
    The parentheses indicate the grouping of input blocks corresponding to the structure of the concatenated protocols.
    }
    \label{table:QBCHEfficiency}
\end{table*}

\section{Numerical Simulation}
\label{sec:numerical}
To evaluate the performance gains of our proposed method, we conducted circuit-level noise simulations using \texttt{Stim}~\cite{stim} to verify the fault-tolerant nature of the optimized protocol for the quantum BCH $[[31,11,5]]$ code.
We adopted a standard circuit-level noise model where each operation is subject to stochastic errors characterized by a physical error rate $p$.
Specifically, each physical CNOT gate is followed by a two-qubit Pauli error, uniformly drawn from the set $\{I,X,Y,Z\}^{\otimes 2}\setminus\{I\otimes I\}$, with a probability of $p/15$ for each of the $15$ nontrivial cases.
Single-qubit gates are subjected to depolarizing noise with probability $p$, while state preparation and measurement (SPAM) operations are subject to bit-flip errors with the same probability $p$.
Under this noise model, we assessed the performance of the 2-to-1 protocol using the operations $((I,R^6),(R^{12},F))$ in Table~\ref{table:QBCHEfficiency} by comparing its logical error rate and acceptance probability with those of the standard 3-to-1 approach.

Figures~\ref{fig:weight_distribution_x} and \ref{fig:weight_distribution_z} show the weight distributions of residual $X$ and $Z$ errors in the output states, where the results for the 2-to-1 and 3-to-1 protocols are plotted in blue and red, respectively. 
In these figures, the probability of errors with weights 1, 2, and $\ge 3$ are presented from top to bottom. 
The numerical data clearly follow the scaling of $O(p)$, $O(p^2)$, and $O(p^3)$, respectively, confirming that the protocols satisfy the definition of strict fault tolerance.

To quantitatively evaluate the error rates, we categorize the residual output errors into two types: transversal errors generated within the distillation circuit and non-malignant errors that propagate from non-FT components and remain undetected. 
Theoretical analysis (see Appendix~\ref{sec:appendix_mto1}) shows that, to leading order in $p$, the effective transversal error rates are given by $4mp/15$ for X errors and $4p/15$ for Z errors, with higher-order contributions neglected.
These theoretical predictions are represented by the black solid lines in the figures.

The close agreement between the numerical data and the theoretical lines indicates that the output error rate is dominated by transversal errors. 
This further demonstrates that our optimization effectively reduces the $X$-error rate by $4p/15$ per reduced input state, ensuring that the reduction in ancilla blocks does not compromise the overall error suppression performance of the protocol. 

Figure~\ref{fig:acceptance_rate} compares the success probability per round for the conventional $3$-to-$1$ protocol and the proposed $2$-to-$1$ protocol.
The dashed and dotted lines represent theoretical lower bounds on the success probability at each step.
These bounds are obtained by evaluating the probability that no errors occur at any circuit location capable of triggering a rejection.
In deriving the lower bound for the second step, we assume that all $X$-errors detectable in the first step have been perfectly identified.
A detailed derivation of these expressions is provided in Appendix~\ref{sec:appendix_derivation}.
Although the simulation data points lie slightly above these bounds due to undetected non-malicious errors that do not trigger rejection, they remain consistent with the theoretical predictions.

Crucially, the proposed $2$-to-$1$ scheme exhibits a consistently higher success probability than the conventional $3$-to-$1$ protocol over the entire range of physical error rates $p$ examined.
This improvement is primarily driven by the reduced circuit complexity of the proposed protocol.
Specifically, by requiring fewer ancilla qubits and transversal CNOT gates per distillation round, the total number of physical error locations is significantly minimized.
This reduction directly suppresses the accumulation of errors that lead to syndrome inconsistencies, thereby substantially increasing the overall round acceptance rate.

\begin{figure}[htb]
  \includegraphics[width=0.5\textwidth]{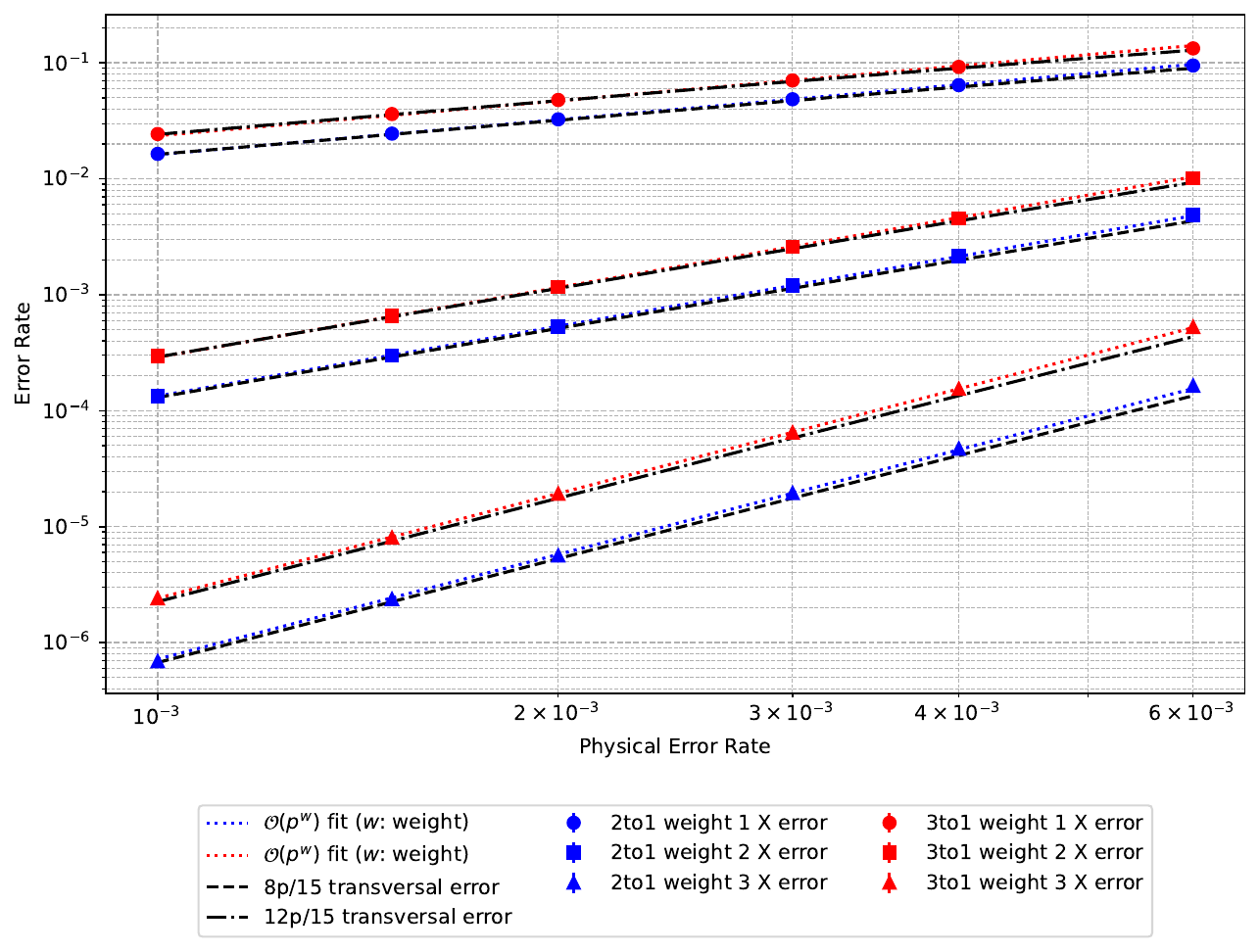}
  \caption{
  Weight distribution of residual $X$ errors
  in the output state $\ket{0_L^{\otimes 11}}$.
  The proposed protocol significantly suppresses
  low-weight $X$ errors compared to the conventional method.
  The black solid line represents the theoretical prediction
  $4mp/15$ for an $m$-to-$1$ protocol.
  }
  \label{fig:weight_distribution_x}
\end{figure}

\begin{figure}[htb]
  \includegraphics[width=0.5\textwidth]{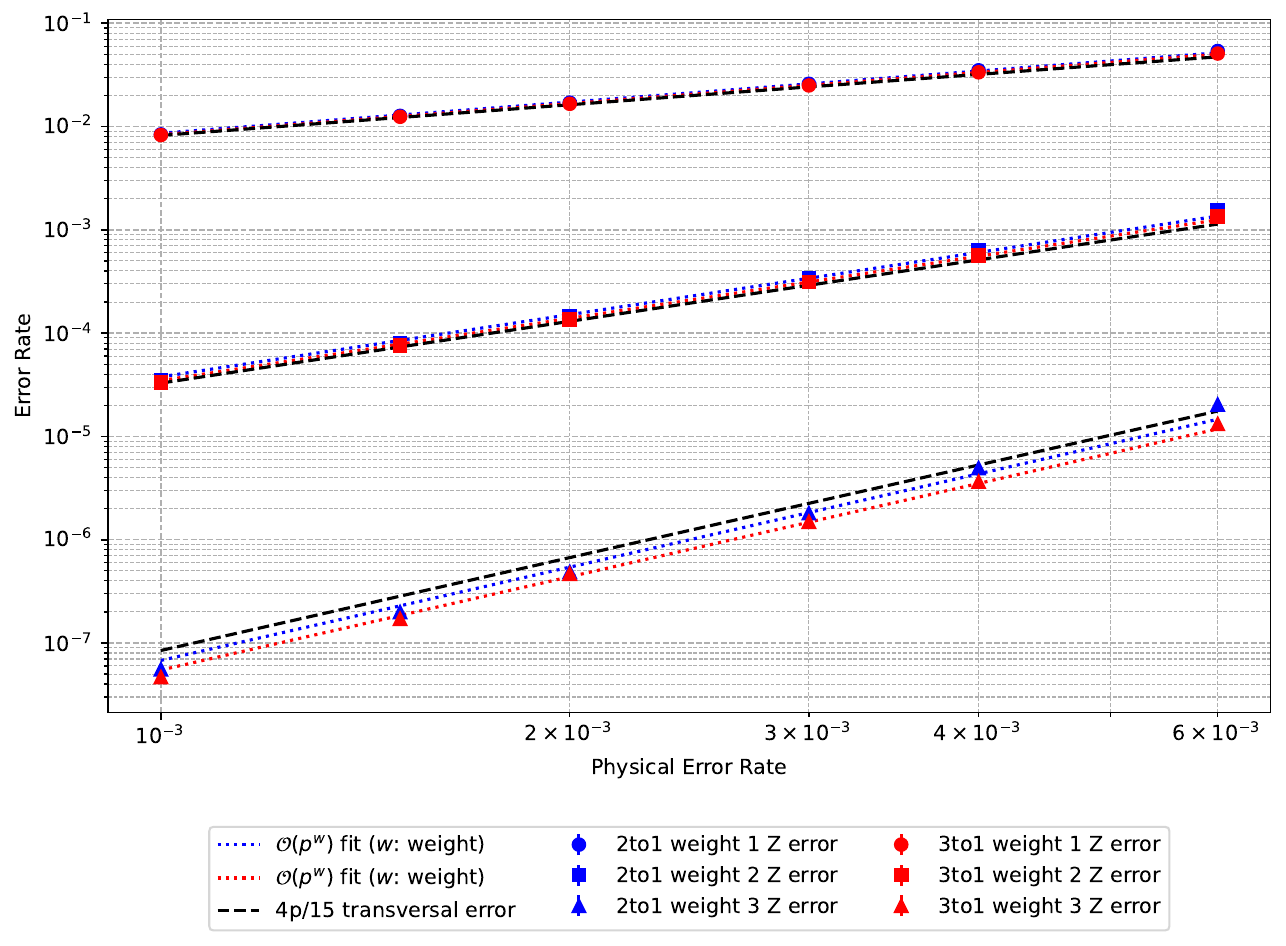}
    \caption{
    Weight distribution of residual $Z$ errors in the output state $\ket{0_L^{\otimes 11}}$.
    The results align with an effective $Z$-error rate of $4p/15$, indicating independence from the distillation protocol.
    }
  \label{fig:weight_distribution_z}
\end{figure}

\begin{figure}[htb]
  \includegraphics[width=0.5\textwidth]{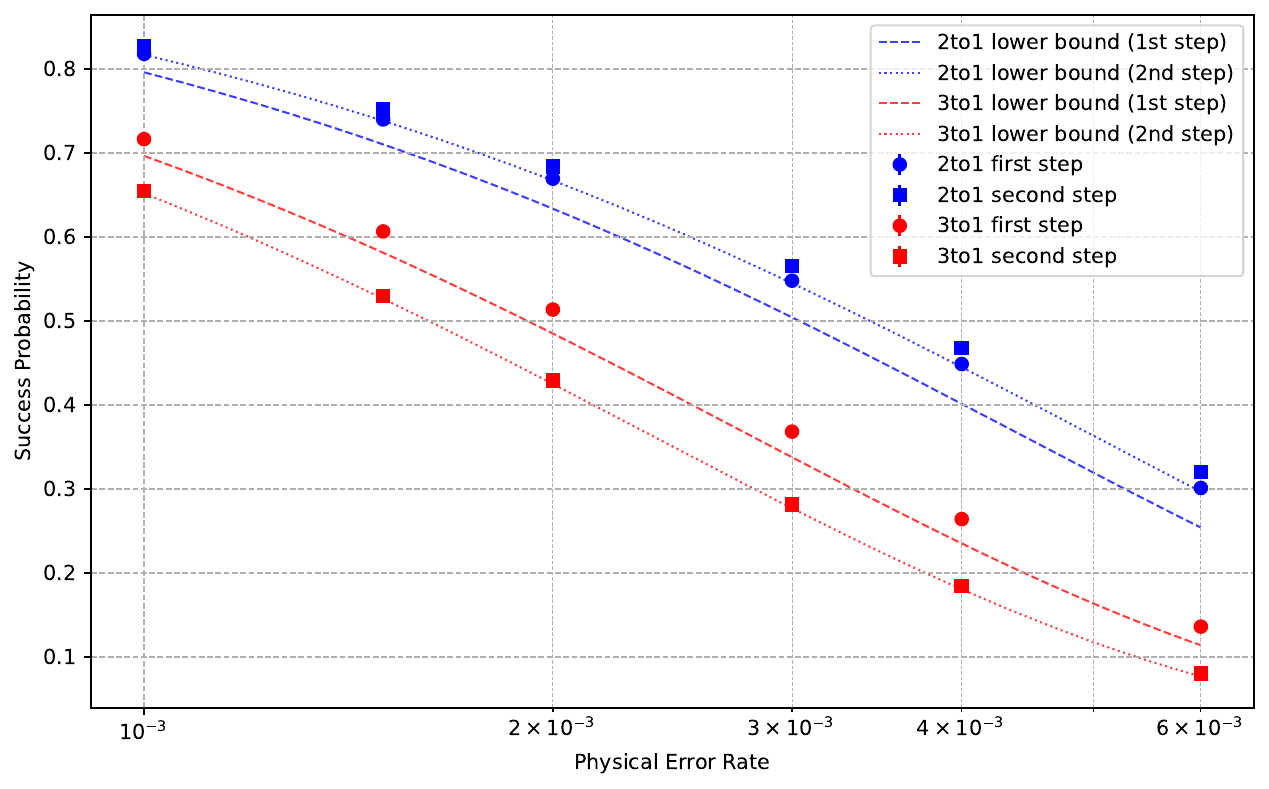}
    \caption{
    Success probability per distillation round for the conventional $3$-to-$1$ and proposed $2$-to-$1$ protocols for the quantum BCH [[31,11,5]] code.
    Red (blue) points correspond to the $3$-to-$1$ ($2$-to-$1$) protocol, while circle (triangle) markers represent the first (second) step.
    The dashed and dotted lines indicate the theoretical lower bounds of the success probability for each step, calculated based on the probability that no errors occur at any detectable locations (see Appendix~\ref{sec:appendix_derivation} for details).
    The proposed scheme achieves a higher success probability due to the effective suppression of error events.
    }
  \label{fig:acceptance_rate}
\end{figure}

\section{Threshold Analysis}
\label{sec:threshold_analysis}
In this section, we conduct a threshold analysis for the transversal CNOT gate, which is a cornerstone of fault-tolerant quantum computing (FTQC), specifically focusing on the performance of quantum BCH codes using our proposed initial state preparation method.

\subsection{Model and Noise Assumptions}
The circuit configuration utilized in this analysis is depicted in Figure~\ref{fig:circuit_config}. % 図のラベルに合わせて変更してください
Each cycle consists of a transversal CNOT gate applied to the initial states, followed by error correction using the error-correcting teleportation.
Here, we evaluate the expected logical error rate per cycle.
In this paper, a logical error is defined as the failure of at least one logical qubit within the code block.

The noise model adopts the parameters established in the numerical simulations in Section~\ref{sec:numerical}.
However, noise associated with the generated ancilla states ($|0_L^{\otimes k}\rangle, |+_L^{\otimes k}\rangle$) is treated as modeled transversal errors.
These errors primarily stem from two sources:
(i) errors induced by transversal CNOT gates in the entanglement distillation circuit, and
(ii) undetectable errors arising from the non-fault-tolerant preparation circuit.
Based on the weight distributions in Figure~\ref{fig:weight_distribution_x} and Figure~\ref{fig:weight_distribution_z}, the contribution of the latter is assumed to be negligible.
When an $m$-to-1 protocol is employed for $X$- and $Z$-error verification in the encoding circuit, the effective noise on the prepared ancilla states can be approximated as follows.
For $\ket{0_L^{\otimes k}}$, the physical $X$- and $Z$-error probabilities are given by $\frac{4m}{15}p$ and $\frac{4}{15}p$, respectively, while for $\ket{+_L^{\otimes k}}$ they are $\frac{4}{15}p$ and $\frac{4m}{15}p$.
Details of this approximation are provided in the Appendix~\ref{sec:appendix_mto1}.

\begin{figure}[htb]
  \includegraphics[width=0.5\textwidth]{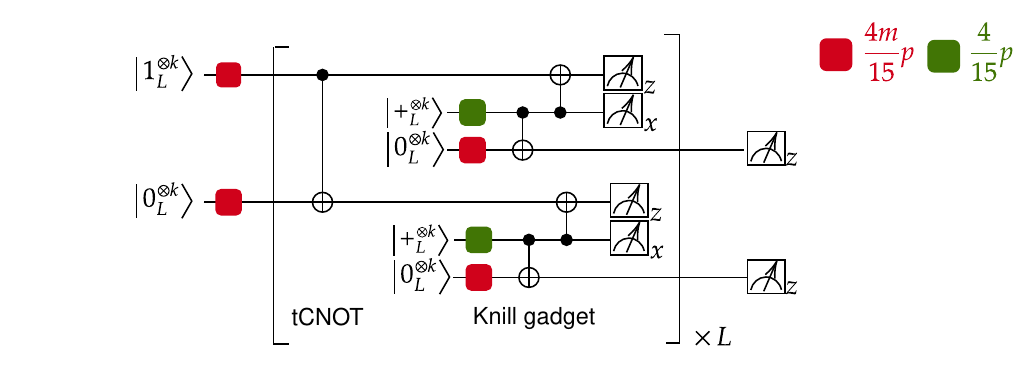} % width=\textwidth でOK（ページ幅いっぱいになる）
  \caption{Circuit configuration for the threshold analysis using $m$-to-1 $X/Z$-error verification in the entanglement distillation circuit.}
  \label{fig:circuit_config}
\end{figure}

\subsection{Analytical Methodology: Scaling Threshold}

%修正後(2026/05/08)
% memo:scaling thresholdの議論をどこまですべきか判断しかねたので，とりあえず詳細に記述．後から削る方針で．
% 第1段落：circuit levelの問題をcode capacityの問題に帰着
The transversality of the circuit allows us to reduce the circuit-level noise analysis to a code-capacity noise model, in which the physical error rate $p$ is upper-bounded by an effective transversal error rate $\gamma p$.
The coefficient $\gamma$ is defined in Appendix~\ref{sec:appendix_gamma}.
We therefore evaluate the logical error rate of each code under this effective noise model.

% 第2,3,4段落：使用するscaling thresholdの説明
Let $t$ denote the number of correctable errors.
Since a logical failure requires at least $t+1$ faults, the logical error rate can be expressed as
\[
    f(p)=a p^{t+1}+O(p^{t+2}),
\]
where $a$ is the leading-order coefficient.
Thus, in the low-error-rate regime, the logical error rate is primarily determined by the leading-order term $a p^{t+1}$.

For finite-size fault-tolerant constructions, a widely used performance metric is the pseudo-threshold, which is defined as the crossing point between the logical error rate $f(p)$ and the physical error rate $p$.
While useful as an overall performance metric, the pseudo-threshold can be affected by higher-order terms near the crossing point.
As a result, it may not accurately reflect the logical error rate in the low-error-rate regime, where one often wants to compare codes at a given physical error rate.

To characterize asymptotic logical error suppression at smaller physical error rates more directly, we introduce the scaling threshold $p_0$ as the crossing point between the leading-order approximation $a p^{t+1}$ and the physical error rate $p$.
It is given by
\[
    a p_0^{t+1}=p_0,
    \qquad
    p_0=a^{-1/t}.
\]
This metric captures the leading-order suppression of logical errors and provides a simple way to compare finite-size codes in the low-error-rate regime.
For the quantum BCH codes studied here, it yields a more conservative estimate than the corresponding pseudo-threshold.

% 第5段落：scaling thresholdを求める際の設定（MLDとPoltyrev bound）
To assess the ultimate potential of quantum BCH codes, we assume Maximum Likelihood Decoding (MLD).
Due to the computational intractability of direct MLD simulations for large codes, we employ the Poltyrev bound~\cite{poltyrev}.
This bound provides a tight upper bound on the logical error rate based on the code's weight distribution.
Using weight distribution data from the Online Encyclopedia of Integer Sequences (OEIS)~\cite{oeis_weight}, we calculated the thresholds for all quantum BCH codes of length $n = 2^m - 1$.

\subsection{Results and Discussion}
We analyze the threshold behavior of quantum BCH codes for all codes with length $n \le 127$ and selected codes with $n=255$.
Figure~\ref{fig:threshold} shows multiple data points for each code, corresponding to different levels of distillation-circuit optimization achieved by the systematic application of $m$-to-1 protocols at each step.
Higher optimization levels consistently yield improved thresholds, owing to the reduced number of CNOT gates in the entanglement distillation circuit, which is the dominant source of output errors.

The results indicate that quantum BCH codes achieve a favorable balance between code rate and error tolerance with realistic physical qubit counts.
In particular, codes such as $\llbracket 31,11,5 \rrbracket$, $\llbracket 63,27,7 \rrbracket$, $\llbracket 127,71,9 \rrbracket$, and $\llbracket 255,159,13 \rrbracket$ exhibit thresholds comparable to that of the Steane code, while their code rates increase by a factor of two to four as the code distance grows.

Across the evaluated codes, we observe a clear trade-off between code rate and threshold.
Within this trade-off, quantum BCH codes span a wide and flexible parameter region (see Figure~\ref{fig:threshold}), reflecting the intrinsic tunability of the BCH family in terms of block length and error-correcting capability.
These threshold results therefore provide practical guidance for designing high-rate fault-tolerant protocols.
Detailed numerical values are summarized in Table~\ref{tab:threshold_analysis} in the Appendix.

\section{Conclusion and Outlook}
In this work, we proposed an optimization method for efficient fault-tolerant preparation of logical Pauli states in quantum BCH codes by exploiting their cyclic symmetry.
This approach generally reduces the required number of ancilla qubits while improving both logical error rates and acceptance probabilities. 
When applied to specific quantum BCH codes, we succeeded in reducing the number of ancilla qubits by more than half, and our circuit-level noise simulations confirmed that these optimized protocols achieve the performance improvements predicted by our theoretical analysis.

In addition, we analyzed the scaling thresholds of the quantum BCH family up to $n=255$, demonstrating that these codes achieve a favorable balance between code rate and error tolerance. 
Specifically, codes like $[[255,159,13]]$ maintain thresholds comparable to the Steane code while increasing code rates by a factor of two to four as the distance grows. 
Our circuit-level optimization consistently improves these thresholds by reducing error accumulation during state preparation, providing practical guidance for designing high-rate and resource-efficient fault-tolerant protocols.

Several directions remain for future research to further enhance the practicality of quantum BCH codes. 
First, the effectiveness of our optimization depends on the structure of the initial non-FT circuits. 
Future studies should address the trade-off between gate count and circuit depth to mitigate idling errors, particularly through the development of automated compilers tailored to specific hardware constraints such as qubit movement costs in neutral-atom systems.

Second, while this study focused on the construction and distillation of these codes, the development of efficient fault-tolerant logical gates for the quantum BCH family is a crucial next step. 
Establishing a complete set of universal operations is essential for integrating these high-rate codes into large-scale quantum computing architectures. 
By addressing these architectural and operational challenges, quantum BCH codes can serve as a robust foundation for scalable and efficient quantum error correction.

\begin{figure*}[htb]
  \includegraphics[width=\textwidth]{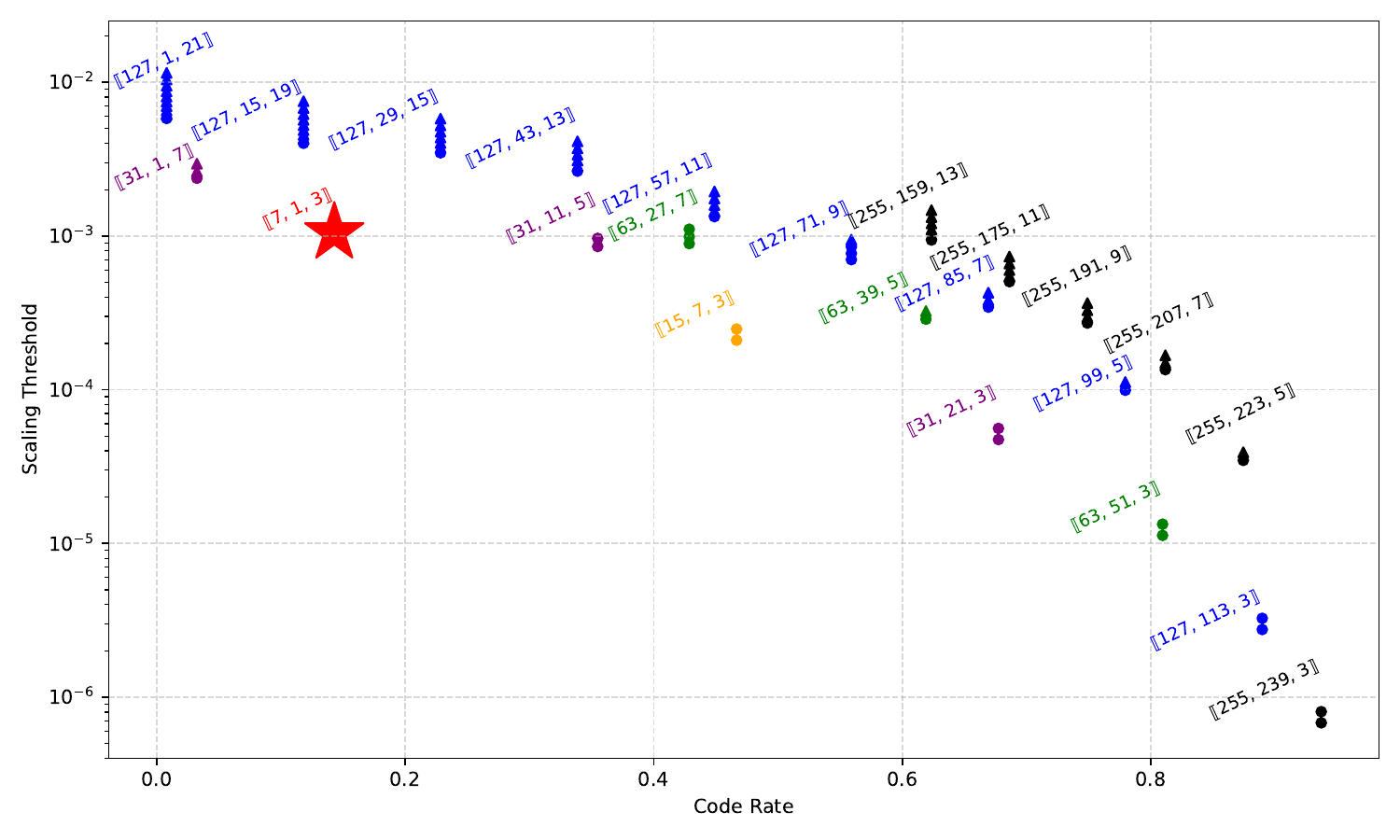} % width=\textwidth でOK（ページ幅いっぱいになる）
  \caption{Threshold analysis results for quantum BCH codes.
Multiple points are plotted for each code, where the threshold systematically improves from the conventional 
($t+1$)-to-$1$ protocol (bottom) to the optimized $2$-to-$1$ protocol (top).}
  \label{fig:threshold}
\end{figure*}

\begin{acknowledgments}
  This work is supported by MEXT Quantum Leap Flagship Program (MEXT Q-LEAP) Grant No. JPMXS0120319794, JST COI-NEXT Grant No. JPMJPF2014, JST Moonshot R\&D Grant No. JPMJMS2061, and JST CREST JPMJCR24I3.
\end{acknowledgments}

\bibliography{ref_bib}

%%%%%%%%%%%%%%%%%%%%%%%%%%%%%%%%%%%%%%%%%%%%%%%%%%%%%%%%%%%%%%%%%%%%%%%%%%%%%%%%%%%%%%%%%%%%%%%%%%%%%%%%%%%%%%%%%%%%%%%%%%%%%%%%%%%%%%%%%%%%%%%%%%%%%%%%%%%%%%%%%%%%%%%%%%%%%%
\appendix
\section{Effective Error Model of the Entanglement Distillation Circuit}
\label{sec:appendix_mto1}

In this appendix, we analyze the errors appearing at the output of the encoding circuits using an \emph{effective error model}, as illustrated in Fig.~\ref{fig:effective_error_model}.
This model allows us to evaluate the effective physical error rates by tracking how Pauli errors generated inside the circuit propagate to the output.
Specifically, Pauli errors arising from faulty gates and measurements are propagated through the circuit according to the standard error propagation rules, and their contributions are accumulated to obtain the effective error rates at the circuit output.

We first consider the distillation circuit of the logical zero state, where the $m_x$-to-1 and $m_z$-to-1 protocols are used for the X- and Z-error checks, respectively. 
The effective X- and Z-error rates of the output state are given by
\begin{equation}
  p_X = \frac{4 m_z}{15} p, \qquad
  p_Z = \frac{4}{15} p .
\end{equation}
The effective X-error rate $p_X$ is obtained by summing the contributions from Pauli error propagation through the transversal CNOT (tCNOT) gates in the distillation circuit.
X-type errors arise from the final $m_z$ tCNOT gates in the first step and from $m_z - 2$ tCNOT gates in the second step, each contributing $\{X \otimes I,\, X \otimes X\}$ with probability $2p/15$.
In addition, the final tCNOT gate in the second step contributes $\{X \otimes I,\, X \otimes X,\, Y \otimes I,\, Y \otimes X\}$ with probability $4p/15$.
In contrast, the effective Z-error rate $p_Z$ originates solely from the final tCNOT gate in the second step, where Pauli errors of the form $\{Y \otimes I,\, Y \otimes X,\, Z \otimes I,\, Z \otimes X\}$ propagate to the output and determine $p_Z$, rendering it independent of the specific distillation protocol.

The logical plus state can be analyzed in the same manner.
It is generated by preparing a logical zero state using a non-fault-tolerant circuit, followed by a transversal Hadamard gate and the entanglement distillation circuit, in which the Z-error check is performed prior to the X-error check.
The resulting effective error rates are given by
\begin{equation}
  p_X = \frac{4}{15} p, \qquad
  p_Z = \frac{4 m_x}{15} p .
\end{equation}

\begin{figure}[htb]
  \includegraphics[width=0.5\textwidth]{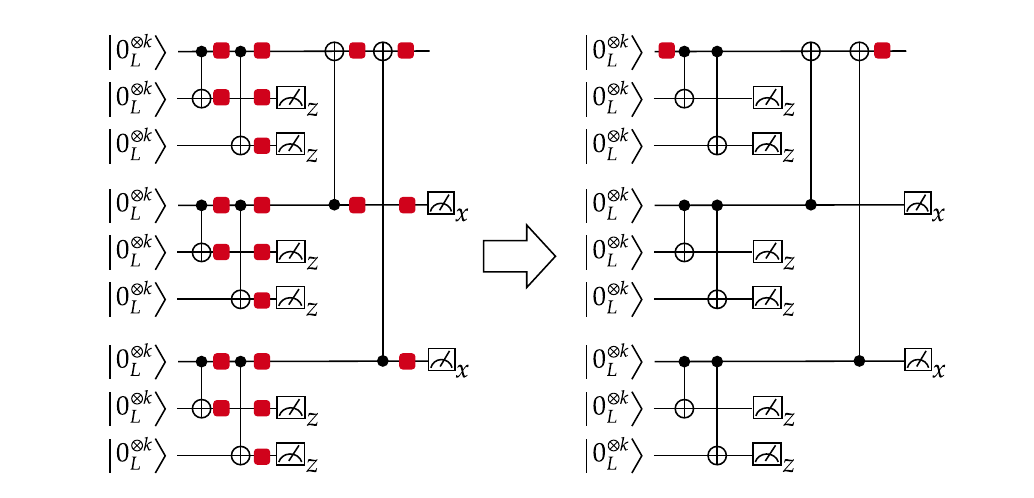}
  \caption{Effective error model for the entanglement distillation circuit.}
  \label{fig:effective_error_model}
\end{figure}

\section{Approximation Method for Circuits Used in the Threshold Analysis}
\label{sec:appendix_gamma}
In the threshold analysis presented in Sec.~\ref{sec:threshold_analysis},
we exploit the fact that the errors occurring in the circuit shown in Fig.~\ref{fig:circuit_config} are transversal.
By approximating the circuit-level noise model with an effective
code-capacity noise model,
the simulation cost can be significantly reduced.
In this appendix, we describe the concrete approximation procedure in detail.

Before analyzing the error propagation in the circuit,
we summarize how transversal errors are treated in our analysis.
Throughout this appendix, we consider only transversal operations
and transversal errors acting on a quantum error-correcting code
$\llbracket n,k,d=2t+1 \rrbracket$.
We first state a lemma concerning the composition of transversal errors.

\begin{lem}
As illustrated below, suppose that bit-flip errors with probabilities
$\alpha p$ and $\beta p$ occur sequentially.
Then, the resulting bit-flip error probability is upper bounded by
$(\alpha+\beta)p$.
\begin{figure}[H]
  \centering
  \includegraphics[width=0.5\textwidth]{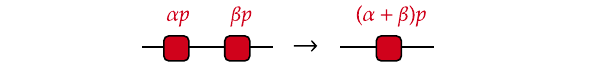}
  \label{fig:appendix_2_1}
\end{figure}
\end{lem}

\begin{proof}
Since the errors are transversal, they act independently on each physical
qubit. It is therefore sufficient to analyze a single physical qubit.
If two bit-flip errors occur sequentially with probabilities $\alpha p$ and
$\beta p$, respectively, the output bit-flip probability is
\[
\alpha p(1-\beta p) + (1-\alpha p)\beta p
= (\alpha+\beta)p - 2\alpha\beta p^2
\leq (\alpha+\beta)p .
\]
Thus, the sequential error process is upper bounded by a single bit-flip
error process with probability $(\alpha+\beta)p$.
\end{proof}

Next, we consider error propagation through a transversal CNOT gate.

\begin{lem}
Consider a two-qubit depolarizing noise acting on the CNOT gate
as illustrated below with probability $p$.
Focusing on $X$ errors,
the marginal probability of an $X$ error occurring
at each of the locations $a$, $b$, and $c$
is exactly $4p/15$.
Moreover, the bit-flip error model obtained by treating these locations
as independent provides an upper bound on the original error probability.
\begin{figure}[H]
  \includegraphics[width=0.5\textwidth]{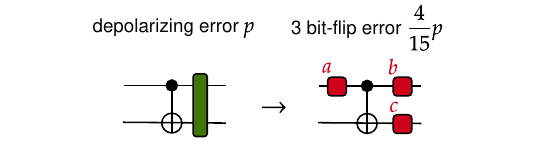}
  \caption{Decomposition of two-qubit depolarizing noise, focusing only on X errors.}
  \label{fig:decompotion}
\end{figure}
\end{lem}

\begin{proof}
The two-qubit depolarizing noise uniformly applies
one of the $15$ nontrivial two-qubit Pauli errors.
Among them, the following $12$ Pauli errors induce $X$ errors:
\begin{align}
\label{eq:pauli_a}
&X\otimes I,\; X\otimes Z,\; Y\otimes I,\; Y\otimes Z, \\
\label{eq:pauli_b}
&I\otimes X,\; I\otimes Y,\; Z\otimes X,\; Z\otimes Y, \\
\label{eq:pauli_c}
&X\otimes X,\; X\otimes Y,\; Y\otimes X,\; Y\otimes Y .
\end{align}
The Pauli errors in Eqs.~\eqref{eq:pauli_a}, \eqref{eq:pauli_b}, and \eqref{eq:pauli_c}
induce single-qubit bit-flip errors
at locations $a$, $b$, and $c$, respectively, as shown in Fig.~\ref{fig:decompotion}.
Since there are four such Pauli errors for each location,
the probability of a bit-flip error at each location is $4p/15$.

We now introduce an effective error model
in which independent bit-flip errors occur
at each location with probability $4p/15$.
This independence assumption allows error patterns
that are not present in the original depolarizing noise,
such as simultaneous independent errors at multiple locations.
As a result, the effective error probability at each location
becomes $4p/15 + (4p/15)^2$.
Therefore, this model uniformly upper bounds
the error probability induced by the original depolarizing noise.
\end{proof}
Using the above lemmas, we obtain the following theorem.
\begin{thm}
For the circuit in Fig.~\ref{fig:circuit_config}, assume two-qubit depolarizing noise $p$ after each CNOT, and bit-flip errors on logical $|0\rangle$, $|+\rangle$ preparations and measurements with probabilities $\alpha p$, $\beta p$, and $p$, respectively. The logical error probability is upper-bounded by a code-capacity model with independent transversal error probability $\gamma p$. Here, $\gamma = ( A^{t+1} + B^{t+1} - C^{t+1} - D^{t+1} )^{1/(t+1)}$ with
\begin{equation}
\label{eq:def_gamma_terms}
\begin{aligned}
A &= 8c_1 + \alpha + 2\beta + 1, \\
B &= 10c_1 + 2\alpha + 3\beta + 1, \\
C &= 3c_1 + \alpha + \beta, \quad D = c_1 + \beta,
\end{aligned}
\end{equation}
where $c_1 = 4/15$.
\end{thm}
\begin{proof}
By Lemmas~1 and~2, physical errors in the circuit of Fig.~\ref{fig:circuit_config} can be mapped to equivalent transversal errors, as shown in Fig.~\ref{fig:proof_1}.
We analyze the rounds in which these errors induce logical failures.
Let $A_n$ ($B_n$) denote the event that a logical error occurs during the upper (lower) error-correction step of round $n$.
Let $A_{\mathrm{END}}$ and $B_{\mathrm{END}}$ denote the events that a logical error occurs during the final measurement.
Figure~\ref{fig:proof_2} summarizes the correspondence between physical errors and logical error events.

Using the inclusion--exclusion principle, the total logical error probability $f(p)$ is written as
\begin{align}
f(p)
&=
p\!\left(
\bigcup_{n=1}^{N} A_n
\cup
\bigcup_{n=1}^{N} B_n
\cup A_{\mathrm{END}}
\cup B_{\mathrm{END}}
\right) \notag \\
&=
\sum_{n=1}^{N}
\bigl[
p(A_n) + p(B_n) - p(A_n \cap B_n)
\bigr] \notag \\
&\quad
-
\sum_{n=2}^{N}
\bigl[
p(A_{n-1} \cap A_n)
+ p(B_{n-1} \cap B_n) \notag \\
&\qquad\qquad
- p(A_{n-1} \cap B_{n-1} \cap A_n)
\bigr] \notag \\
&\quad
- p(A_N \cap A_{\mathrm{END}})
- p(B_N \cap B_{\mathrm{END}}) \notag \\
&\quad
+ p(A_{\mathrm{END}})
+ p(B_{\mathrm{END}}).
\label{eq:total_error_expansion}
\end{align}

\begin{figure}[htb]
  \includegraphics[width=0.5\textwidth]{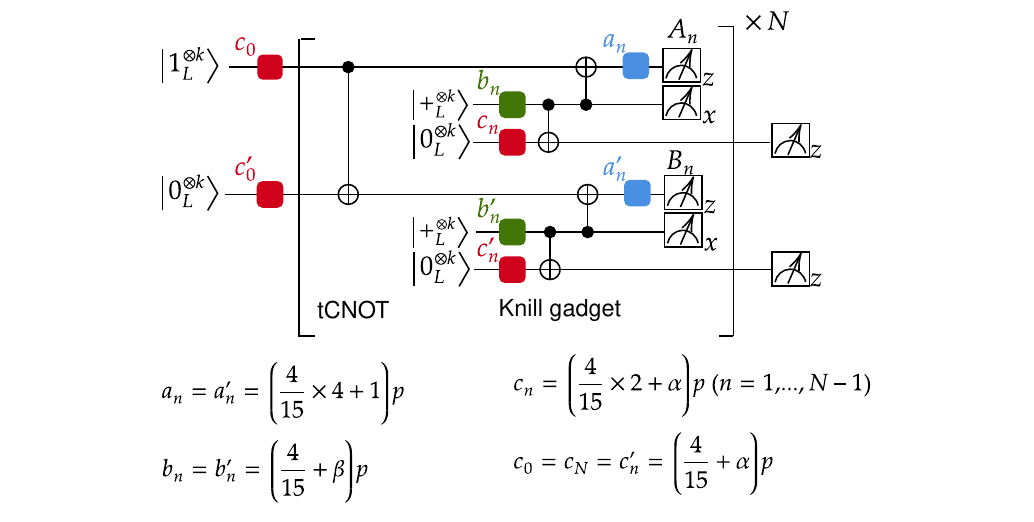}
  \caption{Error decomposition result from Fig.~\ref{fig:circuit_config}.}
  \label{fig:proof_1}
\end{figure}

\begin{figure}[htb]
  \includegraphics[width=0.5\textwidth]{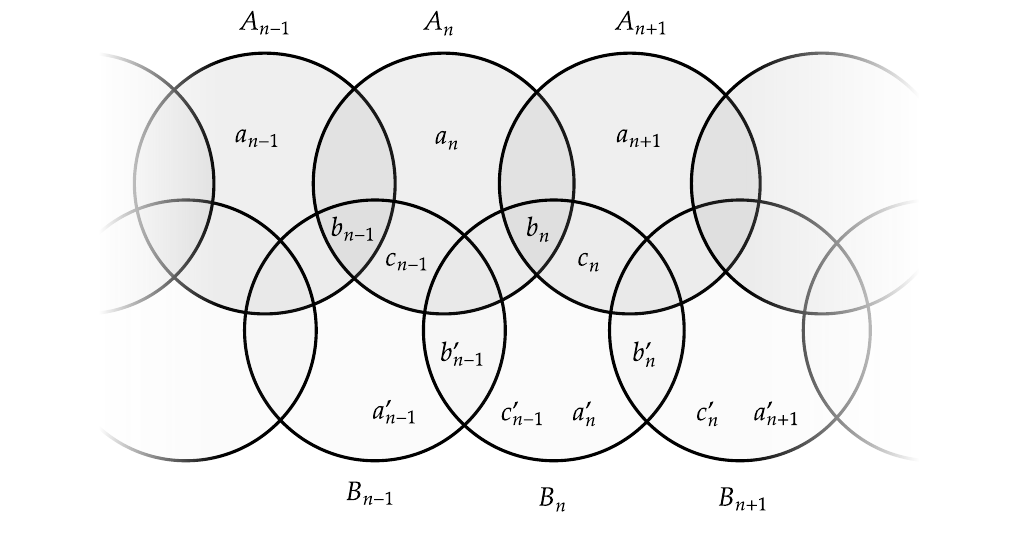}
  \caption{The relationship between each error and error event.}
  \label{fig:proof_2}
\end{figure}

Based on the correspondence shown in Fig.~\ref{fig:proof_2}, we observe that the error mechanisms in the intermediate rounds are translationally invariant.
Therefore, the probabilities of individual error events and their intersections are independent of the round index $n$ for $2 \le n \le N$.
Using the constants defined in Eq.~(\ref{eq:def_gamma_terms}), we can express these first-order physical error probabilities $p_1(\cdot)$as constant values:
\begin{equation}
\begin{aligned}
p_1(A_n) &= Ap, \\
p_1(B_n) &= Bp, \\
p_1(A_n \cap B_n) &= Cp, \\
p_1(A_{n-1} \cap A_n) &= p_1(B_{n-1} \cap B_n) = Dp, \\
p_1(A_{n-1} \cap B_{n-1} \cap A_n) &= Dp.
\end{aligned}
\end{equation}
Note that these relations hold for any $n$ within the range $2 \le n \le N$.

Using these translationally invariant probabilities, the total error probability $f(p)$ can be approximated by a term proportional to $N$ and a boundary term $p_0$ representing the edge effects.
For sufficiently large $N$, the contribution of the boundary term $p_0$ becomes negligible compared to the linear growth of the bulk terms.
Thus, taking the limit of large $N$, the error probability per round converges to:
\begin{align}
\frac{f(p)}{N}
&=
p(A_n) + p(B_n) - p(A_n \cap B_n) \notag \\
&\quad
- p(A_{n-1} \cap A_n) - p(B_{n-1} \cap B_n) \notag \\
&\quad
+ p(A_{n-1} \cap B_{n-1} \cap A_n) \notag \\
\end{align}

For a quantum code with distance $d=2t+1$, the scaling curve of the logical error probability is given by
\begin{equation}
f_0(p) = k p^{t+1},
\end{equation}
where $k$ is a constant.
For arbitrary real coefficients $\lambda_i$ and signatures $s_i \in \{-1, 1\}$,
the relation
\begin{equation}
\sum_i s_i f_0(\lambda_i p)
=
f_0\!\left(
\Bigl(\sum_i s_i \lambda_i^{t+1}\Bigr)^{\frac{1}{t+1}} p
\right)
\end{equation}
holds.
This identity allows multiple independent first-order error sources to be combined into a single effective error rate.

Combining the above results,
the logical error probability of the circuit
can be upper bounded as
\begin{equation}
\frac{f(p)}{N}
\le
f_0(\gamma p).
\end{equation}
The inequality arises because
some logical errors occurring during logical plus-state preparation
do not affect subsequent circuit operations.
The effective error rate $\gamma$
is given by the expression stated in the theorem,
which completes the proof.
\end{proof}

\section{Derivation of Theoretical Lower Bounds for Success Probability}
\label{sec:appendix_derivation}

This section provides the derivation of the theoretical lower bounds for the success probability in each distillation step.
The lower bound is defined as the probability that no errors occur at any potential error source that can be detected during the distillation process.
Figure ~\ref{fig:CNOT_table} summarizes the ancilla preparation workflow and the corresponding noise models used for these derivations.
While we specifically consider the $m$-to-$1$ protocol using the $[[31, 11, 5]]$ quantum BCH code as an example, the calculation is straightforward and can be generalized to other codes by adjusting the circuit parameters.

First, we consider all detectable error sources in the first step, which focuses on $X$-error detection.
The probability that no $X$-errors occur during state preparation and measurement (SPAM), denoted as $P^{(X)}_{\text{SPAM}}$, is determined by the number of sensitive locations in the non-fault-tolerant (non-FT) circuit.
As shown in the non-FT circuit in Figure ~\ref{fig:circuit_31}, the state preparation step begins with $10$ Hadamard gates for the $X$-stabilizers, leading to $21$ potential error locations per ancilla.
Considering $m$ ancilla qubits and the $m-1$ measurement locations where $X$-errors are detectable, the probability is given by
\begin{equation}
P^{(X)}_{\text{SPAM}} = (1 - p)^{21m} (1 - p)^{31(m-1)}.
\end{equation}

Next, we evaluate the probability of no detectable errors within the non-FT circuits, denoted as $P^{(X)}_{\text{non-FT}}$.
For a single-qubit depolarizing error following a Hadamard gate, an $X$-type error (either $X$ or $Y$) occurs with a probability of $2/3p$.
For two-qubit depolarizing errors following CNOT gates, the noise table in Figure ~\ref{fig:CNOT_table} shows that $12$ out of $15$ possible error patterns (indicated by orange cells) contain at least one $X$-error.
Given that the non-FT circuit contains $10$ Hadamard gates and $73$ CNOT gates, and there are $m$ such circuits in the first step, we obtain
\begin{equation}
P^{(X)}_{\text{non-FT}} = \left(1 - \frac{2}{3}p\right)^{10m} \left(1 - \frac{12}{15}p\right)^{73m}.
\end{equation}

Finally, we consider the errors arising from transversal CNOT gates, $P^{(X)}_{\text{trans(1st)}}$.
We distinguish between the last transversal CNOT and the preceding ones because errors on the control qubits of the final CNOT do not need to be considered for $X$-detection.
As shown in the lower-right noise tables of Figure ~\ref{fig:CNOT_table}, the last transversal CNOT contributes $12/15p$ per gate, while the other $m-1$ steps contribute $8/15p$ per gate based on the orange-colored Pauli configurations.
The resulting probability for the $31$ qubits is
\begin{equation}
P^{(X)}_{\text{trans(1st)}} = \left(1 - \frac{12}{15}p\right)^{31} \left(1 - \frac{8}{15}p\right)^{31(m-1)}.
\end{equation}
The total lower bound for the success probability in the first step is thus $P^{(X)} = P^{(X)}_{\text{SPAM}} P^{(X)}_{\text{non-FT}} P^{(X)}_{\text{trans(1st)}}$.

For the second step, we derive the lower bound $P^{(Z)}$ by considering $Z$-error detection.
To simplify the discussion, we assume that any non-malicious errors that passed the first step have a negligible impact on the second step's lower bound.
The error sources in this step are essentially complementary to those in the first step, corresponding to the blue cells in the noise tables of Figure ~\ref{fig:CNOT_table}.
Following the same methodology, the probabilities for each component are estimated as
\begin{equation}
P^{(Z)}_{\text{SPAM}} = (1 - p)^{10m} (1 - p)^{31(m-1)},
\end{equation}
\begin{equation}
P^{(Z)}_{\text{non-FT}} = \left(1 - \frac{1}{3}p\right)^{10m} \left(1 - \frac{3}{15}p\right)^{73m},
\end{equation}
\begin{equation}
P^{(Z)}_{\text{trans(1st)}} = \left(1 - \frac{4}{15}p\right)^{31} \left(1 - \frac{2}{15}p\right)^{31(m-1)},
\end{equation}
\begin{equation}
P^{(Z)}_{\text{trans(2nd)}}=\left(1 - \frac{12}{15}p\right)^{31} \left(1 - \frac{8}{15}p\right)^{31(m-1)}.
\end{equation}
Note that the transversal component for the second step also accounts for the propagation of transversal errors from the first step.
The overall success probability for the second step is then given by $P^{(Z)} = P^{(Z)}_{\text{SPAM}} P^{(Z)}_{\text{non-FT}} P^{(Z)}_{\text{trans(1st)}}P^{(Z)}_{\text{trans(2nd)}}$.

\begin{figure}[htb]
  \includegraphics[width=0.5\textwidth]{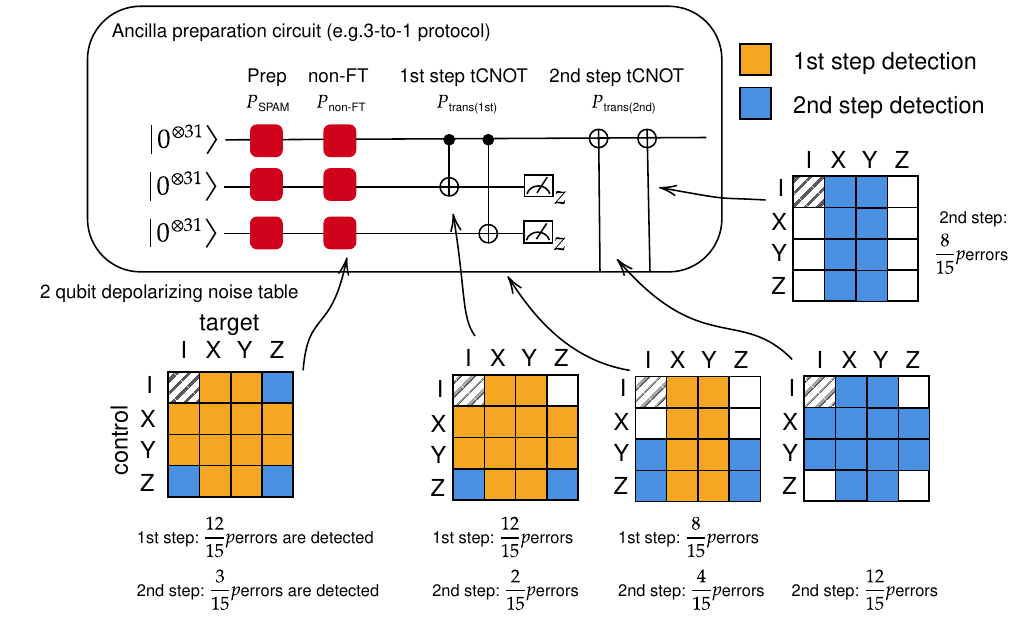}
  \caption{Ancilla preparation circuit and error detection patterns.
The upper panel shows the circuit components for $m$-to-$1$ distillation, with labels $P$ indicating the error sources considered for the lower bound calculation.
The lower panels provide the two-qubit depolarizing noise tables, where cells are color-coded by their detection step (orange: 1st step; blue: 2nd step).
The values $k/15$ denote the detection rates derived from the number of detectable Pauli strings among the $15$ possible two-qubit errors.}
  \label{fig:CNOT_table}
\end{figure}

\clearpage
\onecolumngrid % ★ここから1段組（ページ幅いっぱい）モードに切り替え
\section{Non-FT Circuit Layouts}
\label{sec:appendix_circuit}

% 1つ目の図
\begin{figure}[H] 
  \centering
  \includegraphics[width=\textwidth, height=\textheight, keepaspectratio]{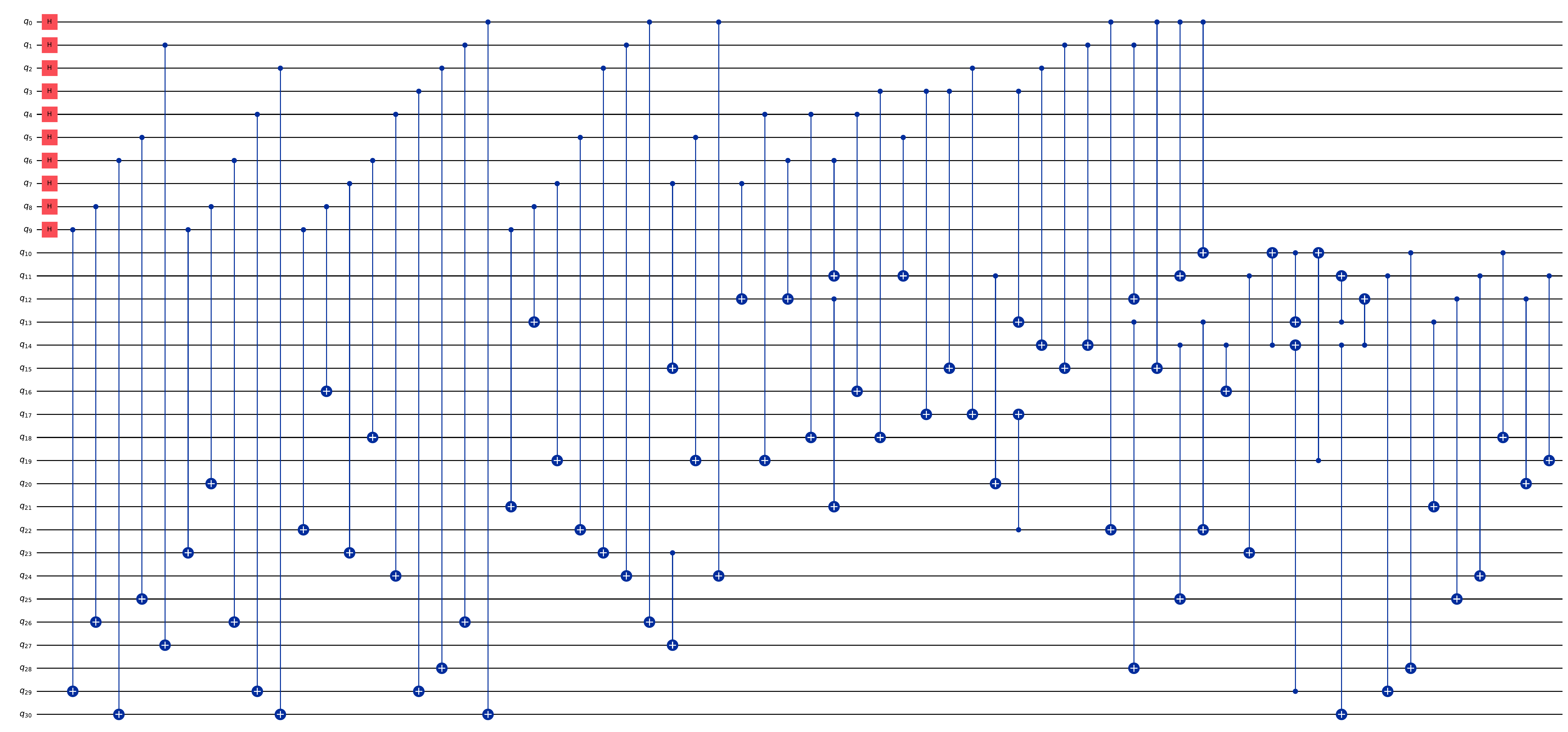}
  % ↑【注意】回転するので、widthにtextheight(縦幅)、heightにtextwidth(横幅)を指定するのがコツです
  \caption{Non-fault-tolerant circuit layout for the quantum BCH code $\llbracket 31,11,5\rrbracket$.}
  \label{fig:circuit_31}
\end{figure}

\clearpage

% 2つ目の図
\begin{sidewaysfigure*}[p]
  \centering
  \includegraphics[width=\textheight, height=\textwidth, keepaspectratio]{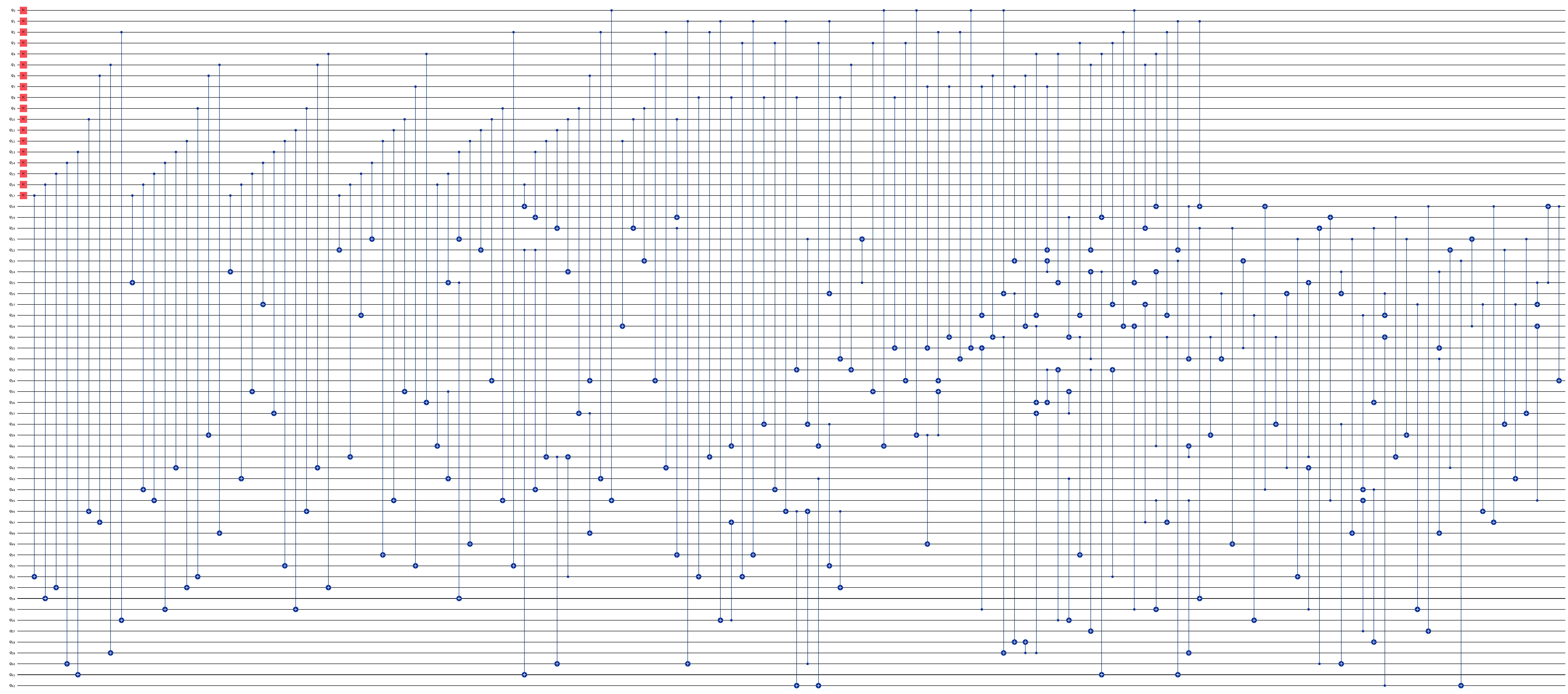}
  \caption{Non-fault-tolerant circuit layout for the quantum BCH code $\llbracket 63,27,7\rrbracket$.}
  \label{fig:circuit_63}
\end{sidewaysfigure*}

\clearpage

% 3つ目の図
\begin{sidewaysfigure*}[p]
  \centering
  \includegraphics[width=0.8\textheight, height=\textwidth, keepaspectratio]{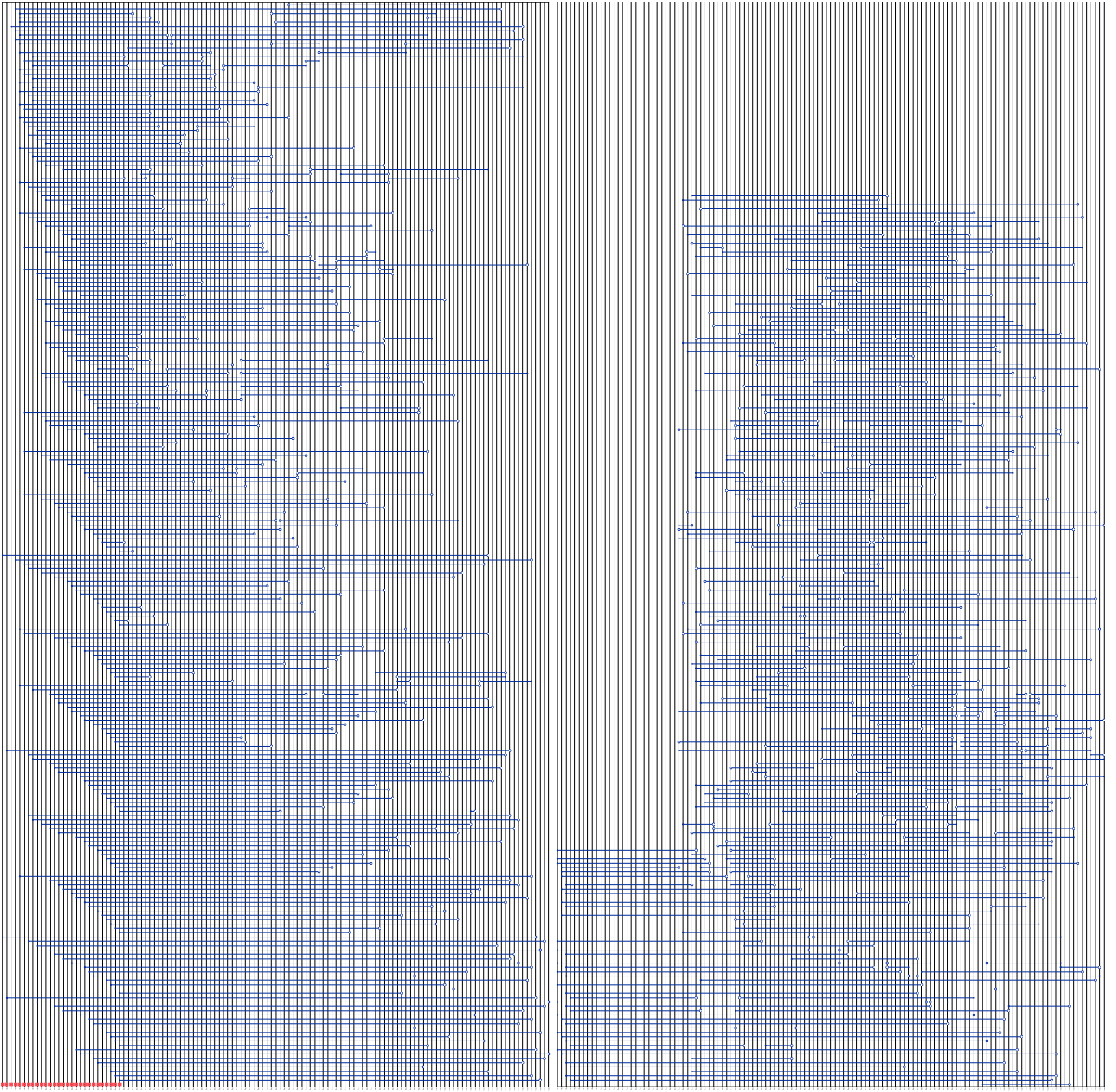}
  \caption{Non-fault-tolerant circuit layout for the quantum BCH code $\llbracket 127,71,9\rrbracket$.}
  \label{fig:circuit_127}
\end{sidewaysfigure*}

\clearpage
% 閾値解析の数値データ
\section{Detailed Performance Tables}
\label{sec:appendix_table}

\begin{longtable}{lcccccc}
\caption{Logical error rates for quantum BCH codes. The column `Distill.' denotes using $m-to-1$ protocol for the ancilla distillation protocol.} \\
\toprule
\makecell[l]{Code \\ $\llbracket n,k,d \rrbracket$} & $t$ & \makecell{Distill. \\ ($m$)} & Rate & $\gamma$ & \makecell{Scaling \\ $p_{th}$} & \makecell{Log. Err. \\ (at $10^{-4}$)} \\
\midrule
\endfirsthead
\multicolumn{7}{c}{{Continued from previous page}} \\
\toprule
\makecell[l]{Code \\ $\llbracket n,k,d \rrbracket$} & $t$ & \makecell{Distill. \\ ($m$)} & Rate & $\gamma$ & \makecell{Scaling \\ $p_{th}$} & \makecell{Log. Err. \\ (at $10^{-4}$)} \\
\midrule
\endhead
\bottomrule
\endfoot
\bottomrule
\endlastfoot
$\llbracket 7,1,3 \rrbracket$ & 1 & 1 & 0.143 & 6.20 & $1.2 \times 10^{-3}$ & $8.1 \times 10^{-6}$ \\
$\llbracket 7,1,3 \rrbracket$ & 1 & 2 & 0.143 & 6.74 & $1.0 \times 10^{-3}$ & $9.5 \times 10^{-6}$ \\
$\llbracket 15,7,3 \rrbracket$ & 1 & 1 & 0.467 & 6.20 & $2.5 \times 10^{-4}$ & $4.0 \times 10^{-5}$ \\
$\llbracket 15,7,3 \rrbracket$ & 1 & 2 & 0.467 & 6.74 & $2.1 \times 10^{-4}$ & $4.8 \times 10^{-5}$ \\
$\llbracket 31,21,3 \rrbracket$ & 1 & 1 & 0.677 & 6.20 & $5.6 \times 10^{-5}$ & $1.8 \times 10^{-4}$ \\
$\llbracket 31,21,3 \rrbracket$ & 1 & 2 & 0.677 & 6.74 & $4.7 \times 10^{-5}$ & $2.1 \times 10^{-4}$ \\
$\llbracket 31,11,5 \rrbracket$ & 2 & 2 & 0.355 & 6.21 & $9.6 \times 10^{-4}$ & $1.1 \times 10^{-6}$ \\
$\llbracket 31,11,5 \rrbracket$ & 2 & 3 & 0.355 & 6.74 & $8.5 \times 10^{-4}$ & $1.4 \times 10^{-6}$ \\
$\llbracket 31,1,7 \rrbracket$ & 3 & 2 & 0.032 & 5.94 & $2.9 \times 10^{-3}$ & $3.9 \times 10^{-9}$ \\
$\llbracket 31,1,7 \rrbracket$ & 3 & 3 & 0.032 & 6.46 & $2.6 \times 10^{-3}$ & $5.5 \times 10^{-9}$ \\
$\llbracket 31,1,7 \rrbracket$ & 3 & 4 & 0.032 & 6.98 & $2.4 \times 10^{-3}$ & $7.5 \times 10^{-9}$ \\
$\llbracket 63,51,3 \rrbracket$ & 1 & 1 & 0.810 & 6.20 & $1.3 \times 10^{-5}$ & $7.5 \times 10^{-4}$ \\
$\llbracket 63,51,3 \rrbracket$ & 1 & 2 & 0.810 & 6.74 & $1.1 \times 10^{-5}$ & $8.9 \times 10^{-4}$ \\
$\llbracket 63,39,5 \rrbracket$ & 2 & 2 & 0.619 & 6.21 & $3.2 \times 10^{-4}$ & $9.5 \times 10^{-6}$ \\
$\llbracket 63,39,5 \rrbracket$ & 2 & 3 & 0.619 & 6.74 & $2.9 \times 10^{-4}$ & $1.2 \times 10^{-5}$ \\
$\llbracket 63,27,7 \rrbracket$ & 3 & 2 & 0.429 & 5.94 & $1.1 \times 10^{-3}$ & $7.4 \times 10^{-8}$ \\
$\llbracket 63,27,7 \rrbracket$ & 3 & 3 & 0.429 & 6.46 & $9.9 \times 10^{-4}$ & $1.0 \times 10^{-7}$ \\
$\llbracket 63,27,7 \rrbracket$ & 3 & 4 & 0.429 & 6.98 & $8.9 \times 10^{-4}$ & $1.4 \times 10^{-7}$ \\
$\llbracket 127,113,3 \rrbracket$ & 1 & 1 & 0.890 & 6.20 & $3.3 \times 10^{-6}$ & $3.1 \times 10^{-3}$ \\
$\llbracket 127,113,3 \rrbracket$ & 1 & 2 & 0.890 & 6.74 & $2.8 \times 10^{-6}$ & $3.6 \times 10^{-3}$ \\
$\llbracket 127,99,5 \rrbracket$ & 2 & 2 & 0.780 & 6.21 & $1.1 \times 10^{-4}$ & $8.0 \times 10^{-5}$ \\
$\llbracket 127,99,5 \rrbracket$ & 2 & 3 & 0.780 & 6.74 & $9.9 \times 10^{-5}$ & $1.0 \times 10^{-4}$ \\
$\llbracket 127,85,7 \rrbracket$ & 3 & 2 & 0.669 & 5.94 & $4.3 \times 10^{-4}$ & $1.3 \times 10^{-6}$ \\
$\llbracket 127,85,7 \rrbracket$ & 3 & 3 & 0.669 & 6.46 & $3.8 \times 10^{-4}$ & $1.8 \times 10^{-6}$ \\
$\llbracket 127,85,7 \rrbracket$ & 3 & 4 & 0.669 & 6.98 & $3.4 \times 10^{-4}$ & $2.4 \times 10^{-6}$ \\
$\llbracket 127,71,9 \rrbracket$ & 4 & 2 & 0.559 & 5.79 & $9.5 \times 10^{-4}$ & $1.2 \times 10^{-8}$ \\
$\llbracket 127,71,9 \rrbracket$ & 4 & 3 & 0.559 & 6.31 & $8.5 \times 10^{-4}$ & $1.9 \times 10^{-8}$ \\
$\llbracket 127,71,9 \rrbracket$ & 4 & 4 & 0.559 & 6.83 & $7.7 \times 10^{-4}$ & $2.9 \times 10^{-8}$ \\
$\llbracket 127,71,9 \rrbracket$ & 4 & 5 & 0.559 & 7.35 & $7.0 \times 10^{-4}$ & $4.1 \times 10^{-8}$ \\
$\llbracket 127,57,11 \rrbracket$ & 5 & 2 & 0.449 & 5.70 & $1.9 \times 10^{-3}$ & $3.6 \times 10^{-11}$ \\
$\llbracket 127,57,11 \rrbracket$ & 5 & 3 & 0.449 & 6.22 & $1.8 \times 10^{-3}$ & $6.1 \times 10^{-11}$ \\
$\llbracket 127,57,11 \rrbracket$ & 5 & 4 & 0.449 & 6.74 & $1.6 \times 10^{-3}$ & $9.9 \times 10^{-11}$ \\
$\llbracket 127,57,11 \rrbracket$ & 5 & 5 & 0.449 & 7.27 & $1.5 \times 10^{-3}$ & $1.5 \times 10^{-10}$ \\
$\llbracket 127,57,11 \rrbracket$ & 5 & 6 & 0.449 & 7.79 & $1.3 \times 10^{-3}$ & $2.4 \times 10^{-10}$ \\
$\llbracket 127,43,13 \rrbracket$ & 6 & 2 & 0.339 & 5.64 & $4.1 \times 10^{-3}$ & $2.0 \times 10^{-14}$ \\
$\llbracket 127,43,13 \rrbracket$ & 6 & 3 & 0.339 & 6.16 & $3.7 \times 10^{-3}$ & $3.8 \times 10^{-14}$ \\
$\llbracket 127,43,13 \rrbracket$ & 6 & 4 & 0.339 & 6.69 & $3.4 \times 10^{-3}$ & $6.7 \times 10^{-14}$ \\
$\llbracket 127,43,13 \rrbracket$ & 6 & 5 & 0.339 & 7.21 & $3.1 \times 10^{-3}$ & $1.1 \times 10^{-13}$ \\
$\llbracket 127,43,13 \rrbracket$ & 6 & 6 & 0.339 & 7.74 & $2.8 \times 10^{-3}$ & $1.9 \times 10^{-13}$ \\
$\llbracket 127,43,13 \rrbracket$ & 6 & 7 & 0.339 & 8.27 & $2.6 \times 10^{-3}$ & $3.0 \times 10^{-13}$ \\
$\llbracket 127,29,15 \rrbracket$ & 7 & 2 & 0.228 & 5.61 & $5.8 \times 10^{-3}$ & $4.6 \times 10^{-17}$ \\
$\llbracket 127,29,15 \rrbracket$ & 7 & 3 & 0.228 & 6.13 & $5.2 \times 10^{-3}$ & $9.3 \times 10^{-17}$ \\
$\llbracket 127,29,15 \rrbracket$ & 7 & 4 & 0.228 & 6.66 & $4.8 \times 10^{-3}$ & $1.8 \times 10^{-16}$ \\
$\llbracket 127,29,15 \rrbracket$ & 7 & 5 & 0.228 & 7.18 & $4.4 \times 10^{-3}$ & $3.3 \times 10^{-16}$ \\
$\llbracket 127,29,15 \rrbracket$ & 7 & 6 & 0.228 & 7.71 & $4.0 \times 10^{-3}$ & $5.9 \times 10^{-16}$ \\
$\llbracket 127,29,15 \rrbracket$ & 7 & 7 & 0.228 & 8.24 & $3.7 \times 10^{-3}$ & $10.0 \times 10^{-16}$ \\
$\llbracket 127,29,15 \rrbracket$ & 7 & 8 & 0.228 & 8.77 & $3.5 \times 10^{-3}$ & $1.6 \times 10^{-15}$ \\
$\llbracket 127,15,19 \rrbracket$ & 9 & 2 & 0.118 & 5.57 & $7.5 \times 10^{-3}$ & $1.3 \times 10^{-21}$ \\
$\llbracket 127,15,19 \rrbracket$ & 9 & 3 & 0.118 & 6.09 & $6.8 \times 10^{-3}$ & $3.2 \times 10^{-21}$ \\
$\llbracket 127,15,19 \rrbracket$ & 9 & 4 & 0.118 & 6.62 & $6.2 \times 10^{-3}$ & $7.4 \times 10^{-21}$ \\
$\llbracket 127,15,19 \rrbracket$ & 9 & 5 & 0.118 & 7.15 & $5.7 \times 10^{-3}$ & $1.6 \times 10^{-20}$ \\
$\llbracket 127,15,19 \rrbracket$ & 9 & 6 & 0.118 & 7.68 & $5.3 \times 10^{-3}$ & $3.3 \times 10^{-20}$ \\
$\llbracket 127,15,19 \rrbracket$ & 9 & 7 & 0.118 & 8.22 & $4.9 \times 10^{-3}$ & $6.4 \times 10^{-20}$ \\
$\llbracket 127,15,19 \rrbracket$ & 9 & 8 & 0.118 & 8.75 & $4.5 \times 10^{-3}$ & $1.2 \times 10^{-19}$ \\
$\llbracket 127,15,19 \rrbracket$ & 9 & 9 & 0.118 & 9.28 & $4.3 \times 10^{-3}$ & $2.2 \times 10^{-19}$ \\
$\llbracket 127,15,19 \rrbracket$ & 9 & 10 & 0.118 & 9.81 & $4.0 \times 10^{-3}$ & $3.8 \times 10^{-19}$ \\
$\llbracket 127,1,21 \rrbracket$ & 10 & 2 & 0.008 & 5.56 & $1.2 \times 10^{-2}$ & $2.5 \times 10^{-25}$ \\
$\llbracket 127,1,21 \rrbracket$ & 10 & 3 & 0.008 & 6.09 & $1.0 \times 10^{-2}$ & $6.7 \times 10^{-25}$ \\
$\llbracket 127,1,21 \rrbracket$ & 10 & 4 & 0.008 & 6.62 & $9.5 \times 10^{-3}$ & $1.7 \times 10^{-24}$ \\
$\llbracket 127,1,21 \rrbracket$ & 10 & 5 & 0.008 & 7.15 & $8.7 \times 10^{-3}$ & $3.9 \times 10^{-24}$ \\
$\llbracket 127,1,21 \rrbracket$ & 10 & 6 & 0.008 & 7.68 & $8.1 \times 10^{-3}$ & $8.6 \times 10^{-24}$ \\
$\llbracket 127,1,21 \rrbracket$ & 10 & 7 & 0.008 & 8.21 & $7.5 \times 10^{-3}$ & $1.8 \times 10^{-23}$ \\
$\llbracket 127,1,21 \rrbracket$ & 10 & 8 & 0.008 & 8.74 & $7.0 \times 10^{-3}$ & $3.6 \times 10^{-23}$ \\
$\llbracket 127,1,21 \rrbracket$ & 10 & 9 & 0.008 & 9.27 & $6.5 \times 10^{-3}$ & $6.9 \times 10^{-23}$ \\
$\llbracket 127,1,21 \rrbracket$ & 10 & 10 & 0.008 & 9.81 & $6.2 \times 10^{-3}$ & $1.3 \times 10^{-22}$ \\
$\llbracket 127,1,21 \rrbracket$ & 10 & 11 & 0.008 & 10.34 & $5.8 \times 10^{-3}$ & $2.3 \times 10^{-22}$ \\
$\llbracket 255,159,13 \rrbracket$ & 6 & 2 & 0.624 & 5.64 & $1.5 \times 10^{-3}$ & $9.9 \times 10^{-12}$ \\
$\llbracket 255,159,13 \rrbracket$ & 6 & 3 & 0.624 & 6.16 & $1.3 \times 10^{-3}$ & $1.8 \times 10^{-11}$ \\
$\llbracket 255,159,13 \rrbracket$ & 6 & 4 & 0.624 & 6.69 & $1.2 \times 10^{-3}$ & $3.2 \times 10^{-11}$ \\
$\llbracket 255,159,13 \rrbracket$ & 6 & 5 & 0.624 & 7.21 & $1.1 \times 10^{-3}$ & $5.5 \times 10^{-11}$ \\
$\llbracket 255,159,13 \rrbracket$ & 6 & 6 & 0.624 & 7.74 & $1.0 \times 10^{-3}$ & $9.1 \times 10^{-11}$ \\
$\llbracket 255,159,13 \rrbracket$ & 6 & 7 & 0.624 & 8.27 & $9.4 \times 10^{-4}$ & $1.4 \times 10^{-10}$ \\
$\llbracket 255,175,11 \rrbracket$ & 5 & 2 & 0.686 & 5.70 & $7.3 \times 10^{-4}$ & $4.7 \times 10^{-9}$ \\
$\llbracket 255,175,11 \rrbracket$ & 5 & 3 & 0.686 & 6.22 & $6.6 \times 10^{-4}$ & $7.9 \times 10^{-9}$ \\
$\llbracket 255,175,11 \rrbracket$ & 5 & 4 & 0.686 & 6.74 & $6.0 \times 10^{-4}$ & $1.3 \times 10^{-8}$ \\
$\llbracket 255,175,11 \rrbracket$ & 5 & 5 & 0.686 & 7.27 & $5.5 \times 10^{-4}$ & $2.0 \times 10^{-8}$ \\
$\llbracket 255,175,11 \rrbracket$ & 5 & 6 & 0.686 & 7.79 & $5.0 \times 10^{-4}$ & $3.1 \times 10^{-8}$ \\
$\llbracket 255,191,9 \rrbracket$ & 4 & 2 & 0.749 & 5.79 & $3.7 \times 10^{-4}$ & $5.6 \times 10^{-7}$ \\
$\llbracket 255,191,9 \rrbracket$ & 4 & 3 & 0.749 & 6.31 & $3.3 \times 10^{-4}$ & $8.6 \times 10^{-7}$ \\
$\llbracket 255,191,9 \rrbracket$ & 4 & 4 & 0.749 & 6.83 & $3.0 \times 10^{-4}$ & $1.3 \times 10^{-6}$ \\
$\llbracket 255,191,9 \rrbracket$ & 4 & 5 & 0.749 & 7.35 & $2.7 \times 10^{-4}$ & $1.8 \times 10^{-6}$ \\
$\llbracket 255,207,7 \rrbracket$ & 3 & 2 & 0.812 & 5.94 & $1.7 \times 10^{-4}$ & $2.1 \times 10^{-5}$ \\
$\llbracket 255,207,7 \rrbracket$ & 3 & 3 & 0.812 & 6.46 & $1.5 \times 10^{-4}$ & $3.0 \times 10^{-5}$ \\
$\llbracket 255,207,7 \rrbracket$ & 3 & 4 & 0.812 & 6.98 & $1.3 \times 10^{-4}$ & $4.1 \times 10^{-5}$ \\
$\llbracket 255,223,5 \rrbracket$ & 2 & 2 & 0.875 & 6.21 & $3.9 \times 10^{-5}$ & $6.5 \times 10^{-4}$ \\
$\llbracket 255,223,5 \rrbracket$ & 2 & 3 & 0.875 & 6.74 & $3.5 \times 10^{-5}$ & $8.3 \times 10^{-4}$ \\
$\llbracket 255,239,3 \rrbracket$ & 1 & 1 & 0.937 & 6.20 & $8.1 \times 10^{-7}$ & $1.2 \times 10^{-2}$ \\
$\llbracket 255,239,3 \rrbracket$ & 1 & 2 & 0.937 & 6.74 & $6.8 \times 10^{-7}$ & $1.5 \times 10^{-2}$ \\
\label{tab:threshold_analysis}
\end{longtable}
\end{document}